\documentclass{article}
\pdfoutput=1
\usepackage{amssymb}

\usepackage{etex}
\usepackage{amsmath}
\usepackage{tikz}
\usepackage{tkz-graph}
\usepackage{bussproofs}
\usepackage{graphicx}

\usetikzlibrary{arrows,automata}
\usetikzlibrary{arrows,shapes,positioning}
\usetikzlibrary{tikzmark,decorations.pathreplacing,calc}
\newtheorem{theorem}{Theorem}

\newtheorem{claim}[theorem]{Claim}
\newtheorem{conclusion}[theorem]{Conclusion}

\newtheorem{corollary}[theorem]{Corollary}

\newtheorem{definition}[theorem]{Definition}
\newtheorem{example}[theorem]{Example}

\newtheorem{lemma}[theorem]{Lemma}

\newtheorem{partial solution}[theorem]{Partial Solution}

\newenvironment{proof}[1][Proof]{\textbf{#1.} }{\ \rule{0.5em}{0.5em}}
\newcommand{\imply}{\rightarrow}
\input{tcilatex}

\begin{document}

\begin{center}
{\Large NP vs PSPACE\smallskip }

L. Gordeev, E. H. Haeusler

\textit{Universit\"{a}t T\"{u}bingen, Ghent University, PUC Rio de Janeiro}

l\texttt{ew.gordeew@uni-tuebingen.de}\textit{,}

\textit{PUC Rio de Janeiro}

\texttt{hermann@inf.puc-rio.br}
\end{center}

\subparagraph{Abstract.}

We present a proof of the conjecture $\mathcal{NP}$ = $\mathcal{PSPACE}$ by
showing that arbitrary tautologies of Johansson's minimal propositional
logic admit ``small'' polynomial-size dag-like natural deductions in
Prawitz's system for minimal propositional logic. These ``small'' deductions
arise from standard ``large''\ tree-like inputs by horizontal dag-like
compression that is obtained by merging distinct nodes labeled with
identical formulas occurring in horizontal sections of deductions involved.
The underlying ``geometric'' idea: if the height, $h\left( \partial \right) $%
, and the total number of distinct formulas, $\phi \left( \partial \right) $%
, of a given tree-like deduction $\partial $\ of a minimal tautology $\rho $
are both polynomial in the length of $\rho $, $\left| \rho \right| $, then
the size of the horizontal dag-like compression $\partial ^{\text{\textsc{c}}%
}$\ is at most $h\left( \partial \right) \times \phi \left( \partial \right) 
$, and hence polynomial in $\left| \rho \right| $. Moreover if maximal
formula length in $\partial $, $\mu \left( \partial \right) $, is also
polynomial in $\left| \rho \right| $ , then so is the weight of $\partial ^{%
\text{\textsc{c}}}$. That minimal tautologies $\rho $ are derivable by
natural deductions $\partial $ with $\left| \rho \right| $-polynomial $%
h\left( \partial \right) $, $\phi \left( \partial \right) $ and $\mu \left(
\partial \right) $ follows via embedding from the known result that there
are analogous sequent calculus deductions of sequent $\Rightarrow \rho $.
The attached proof is due to the first author, but it was the second author
who proposed an initial idea to attack a weaker conjecture $\mathcal{NP}=%
\mathcal{\mathit{co}NP}$ by reductions in diverse natural deduction
formalisms for propositional logic. That idea included interactive use of
minimal, intuitionistic and classical formalisms, so its practical
implementation was too involved. On the contrary, the attached proof of $%
\mathcal{NP}=\mathcal{PSPACE}$ runs inside the natural deduction
interpretation of Hudelmaier's cutfree sequent calculus for minimal logic.


\textbf{Acknowledgments}

This work arose in the context of term- and proof-compression research
supported by the ANR/DFG projects \negthinspace \emph{HYPOTHESES} and \emph{%
BEYOND LOGIC} [DFG grants 275/16-1, 16-2, 17-1] and the CNPq project\emph{\
Proofs: Structure, Transformations and Semantics} [grant 402429/2012-5]. We
would like to thank L. C. Pereira and all colleagues in PUC-Rio for their
contribution as well as P. Schroeder-Heister (EKUT) and M. R. F. Benevides
(UFRJ) for their support of these projects. Special thanks goes to S. Buss,
R. Dyckhoff and F. Gilbert for their insightful comments.

\section{Introduction}

Recall standard definitions of the complexity classes $\mathcal{NP}$, 
\textit{co}$\mathcal{NP}$ and $\mathcal{PSPACE}$. $L\subseteq \left\{
0,1\right\} ^{\ast }$\ is in $\mathcal{NP}$, resp. \textit{co}$\mathcal{NP}$%
, if there exists a polynomial $p$ and a polytime TM $M$ such that 
\begin{equation*}
\begin{array}{c}
\quad \quad \ \fbox{$x\in L\Leftrightarrow \left( \exists u\in \left\{
0,1\right\} ^{p\left( \left| x\right| \right) }\right) \!M\left( x,u\right)
=1$}\text{,} \\ 
\text{resp. }\fbox{$x\in L\Leftrightarrow \left( \forall u\in \left\{
0,1\right\} ^{p\left( \left| x\right| \right) }\right) $\negthinspace $%
M\left( x,u\right) =1$}\text{,}
\end{array}
\end{equation*}
holds for every $x\in \left\{ 0,1\right\} ^{\ast }$. Now $L\subseteq \left\{
0,1\right\} ^{\ast }$\ is in $\mathcal{PSPACE}$ if there exists a polynomial 
$p$ and a TM $M$ such that for every input $x\in \left\{ 0,1\right\} ^{\ast
} $, the total number of non-blank locations that occur during $M$'s
execution on $x$ is at most $p\left( \left| x\right| \right) $, and $x\in
L\Leftrightarrow M\left( x\right) =1$. It is well-known that $\mathcal{NP}$ $%
\subseteq $ $\mathcal{PSPACE}$ and \textit{co}$\mathcal{NP}$ $\subseteq $ $%
\mathcal{PSPACE}$. Moreover, if $\mathcal{NP}$ $=$\ $\mathcal{PSPACE}$ then $%
\mathcal{NP}$ $=$\ \textit{co}$\mathcal{NP}$. The latter conjecture seems
more natural and/or plausible, as it reflects an idea of logical equivalence
between model theoretical (re: $\mathcal{NP}$) and proof theoretical (re: 
\textit{co}$\mathcal{NP}$) interpretations of non-deterministic polytime
computability. So according to familiar NP-(coNP)-completeness of boolean
satisfiability (resp. validity) problem, in order to prove $\mathcal{NP}$ $=$%
\ \textit{co}$\mathcal{NP}$ it will suffice to show that arbitrary
tautologies admit ``small''polynomial-size (abbr.: \emph{polysize})
deductions in a natural propositional proof system. The former (stronger)
conjecture\ $\mathcal{NP}$ $=$\ $\mathcal{PSPACE}$ is less intuitive than $%
\mathcal{NP}$ $=$\ \textit{co}$\mathcal{NP}$, but our proof thereof follows
the same pattern with respect to minimal logic, instead of classical one.
This is legitimate, since the validity of minimal propositional logic is
PSPACE-complete.

\section{Towards $\mathcal{NP}$ $=$\ $\mathcal{PSPACE}$}

\subsection{Proof theoretic background}

We consider two types of proof theoretic formalism: Gentzen-style Sequent
Calculus (abbr.: SC) and Prawitz's Natural Deduction (abbr.: ND). Both SC
and ND admit standard tree-like interpretation, as well as generalized
dag-like interpretation in which proofs (or deductions) are regarded as
labeled rooted monoedge dags. \footnote{{\footnotesize Recall that \emph{%
`dag'} stands for \emph{directed acyclic graph} (edges directed upwards).}}
Our desired ``small'' deductions will arise from ``large''\ standard
tree-like inputs by appropriate dag-like compressing techniques. The
compression in question is obtained by merging distinct nodes with identical
labels, i.e. sequents or single formulas in the corresponding case of SC or
ND, respectively.

In our earlier SC related proof-compression research \cite{JMS}, \cite{IGPL}%
, \cite{Gor} dealing with sequent calculi\emph{\ }\footnote{{\footnotesize %
Also note \cite{Marcela} that shows a mimp-like formalization of natural
deductions that admits ``explicit'' and size-preserving strong normalization
procedure.}} we obtained such basic result (et al):

\emph{Any tree-like deduction }$\partial $\emph{\ of any given sequent }$S$%
\emph{\ is constructively compressible to a dag-like deduction }$\partial ^{%
\text{\textsc{c}}}$\emph{\ of }$S$\emph{\ in which sequents occur at most
once. I.e., in }$\partial ^{\text{\textsc{c}}}$\emph{, distinct nodes are
supplied with distinct sequents (that occur in }$\partial $\emph{).}

However, even in the case of cutfree SC having good proof search and other
nice properties (like Gentzen's subformula property), this result still
gives us no polynomial control over the size of $\partial ^{\text{\textsc{c}}%
}$. The reason is that sequents occurring in $\partial ^{\text{\textsc{c}}}$
can be viewed as collections of subformulas of $S$, which allows their total
number to grow exponentially in the size of $S$, $\left| S\right| $. In
contrast, ND deductions consist of single formulas, which gives hope to
overcome this problem. On the other hand, in ND, full dag-like compression
merging arbitrary nodes supplied with identical formulas is problematic, as
there is a risk of confusion between deduced formulas and the same formulas
used above as discharged assumptions. But we can try \emph{horizontal
dag-like compression} that should merge only the nodes occurring in
horizontal sections of ND deductions involved. The underlying idea is
explained in the abstract. Namely, if a tree-like input deduction $\partial $%
\ of a given formula $\rho $ has $\left| \rho \right| $-polynomial \emph{%
height} (= maximal thread length), $h\left( \partial \right) $, and the 
\emph{foundation} (= the total number of distinct formulas occurring in $%
\partial $), $\phi \left( \partial \right) $, is also polynomial in $\left|
\rho \right| $, then the \emph{size} (= total number of formulas) of the
corresponding horizontal dag-like compression $\partial ^{\text{\textsc{c}}}$%
, $\left| \partial ^{\text{\textsc{c}}}\right| $, will be at most $h\left(
\partial \right) \times \phi \left( \partial \right) $. Moreover if maximal
formula length in $\partial $, $\mu \left( \partial \right) $, is also
polynomial in $\left| \rho \right| $, then the \emph{weight} (= total number
of characters occurring inside) of $\partial ^{\text{\textsc{c}}}$, $\left\|
\partial ^{\text{\textsc{c}}}\right\| $, is bounded by $h\left( \partial
\right) \times \phi \left( \partial \right) \times \mu \left( \partial
\right) $. It remains to show that every formula $\rho $ that is valid in
minimal logic admits a ND deduction $\partial $ with $\left| \rho \right| $%
-polynomial parameters $h\left( \partial \right) $, $\phi \left( \partial
\right) $ and $\mu \left( \partial \right) $. But this follows by a natural
SC $\hookrightarrow $ ND embedding from Hudelmaier's result saying that
there are analogous SC deductions of the corresponding sequent $\Rightarrow
\rho $.

\subsection{Overview of the proof}

We argue as follows along the lines 1--4:

\begin{enumerate}
\item  Formalize minimal propositional logic as fragment \textsc{LM}$%
_{\rightarrow }$ of Hudelmaier's tree-like cutfree intuitionistic sequent
calculus. For any \textsc{LM}$_{\rightarrow }$ proof $\partial $ of sequent $%
\Rightarrow \rho \ $:

\begin{enumerate}
\item  $h\left( \partial \right) $ (= the height) is polynomial (actually
linear) in $\left| \rho \right| $,

\item  $\phi \left( \partial \right) $ (= total number of formulas) and $\mu
\left( \partial \right) $ (= maximal formula length) are also polynomial in $%
\left| \rho \right| $.
\end{enumerate}

\item  Show that there exists a constructive (1)+(2) preserving embedding $%
\mathcal{F}$ of \textsc{LM}$_{\rightarrow }$ into Prawitz's tree-like
natural deduction formalism \textsc{NM}$_{\rightarrow }$ for minimal logic.

\item  Elaborate polytime verifiable dag-like deducibility in \textsc{NM}$%
_{\rightarrow }$.

\item  Elaborate and apply \emph{horizontal tree-to-dag proof compression}
in \textsc{NM}$_{\rightarrow }$. For any tree-like \textsc{NM}$_{\rightarrow
}$ input $\partial $, the weight of dag-like output $\partial ^{\text{%
\textsc{c}}}$ is bounded by $h\left( \partial \right) \times \phi \left(
\partial \right) \times \mu \left( \partial \right) $. Hence the weight of $%
\left( \mathcal{F}\left( \partial \right) \right) ^{\text{\textsc{c}}}$ for
any given tree-like \textsc{LM}$_{\rightarrow }$ proof $\partial $ of $\rho $
is polynomially bounded in $\left| \rho \right| $. Since minimal logic is
PSPACE-complete, conclude that $\mathcal{NP}$ $=$\ $\mathcal{PSPACE}$.
\end{enumerate}

\section{More detailed exposition}

In the sequel we consider standard language $\mathcal{L}_{\rightarrow }$ of
minimal logic whose formulas ($\alpha $, $\beta $, $\gamma $, $\rho $ etc.)
are built up from propositional variables ($p$, $q$, $r$, etc.) using one
propositional connective $\rightarrow $. The sequents are in the form $%
\Gamma \Rightarrow \alpha $\ whose antecedents, $\Gamma $,\ are viewed as
multisets of formulas; sequents $\Rightarrow \alpha $\ , i.e. $\emptyset
\Rightarrow \alpha $, are identified with formulas $\alpha $.

\subsection{Sequent calculus \textsc{LM}$_{\rightarrow }$}

\textsc{LM}$_{\rightarrow }$ includes the following axioms $\left( \text{%
\textsc{M}}A\right) $ and inference rules $\left( \text{\textsc{M}}%
I1\rightarrow \right) $, $\left( \text{\textsc{M}}I2\rightarrow \right) $, $%
\left( \text{\textsc{M}}E\rightarrow P\right) $, $\left( \text{\textsc{M}}%
E\rightarrow \rightarrow \right) $ in the language $\mathcal{L}_{\rightarrow
}$ (the constraints are shown in square brackets). \footnote{{\footnotesize %
This is a slightly modified, equivalent version of the corresponding \
purely implicational and }$\bot ${\footnotesize -free} {\footnotesize %
subsystem of Hudelmaier's intuitionistic calculus \textsc{LG}, cf. \cite
{Hudelmaier}. The constraints }$q\in VAR\left( \Gamma ,\gamma \right) $%
{\footnotesize \ are added just for the sake of transparency.}}\medskip

$\fbox{$\left( \text{\textsc{M}}A\right) :\ \ \ \Gamma ,p\Rightarrow p$}$

$\fbox{$\left( \text{\textsc{M}}I1\!\rightarrow \right) :\ \ \ \dfrac{\Gamma
,\alpha \Rightarrow \beta }{\Gamma \Rightarrow \alpha \rightarrow \beta }%
\smallskip \quad \left[ \left( \nexists \gamma \right) :\left( \alpha
\rightarrow \beta \right) \rightarrow \gamma \in \Gamma \right] $}$

$\fbox{$\left( \text{\textsc{M}}I2\!\rightarrow \right) :\ \ \ \dfrac{\Gamma
,\alpha ,\beta \rightarrow \gamma \Rightarrow \beta }{\Gamma ,\left( \alpha
\rightarrow \beta \right) \rightarrow \gamma \Rightarrow \alpha \rightarrow
\beta }$}$

$\fbox{$\left( \text{\textsc{M}}E\!\rightarrow \!P\right) :\ \ \ \dfrac{%
\Gamma ,p,\gamma \Rightarrow q}{\Gamma ,p,p\rightarrow \gamma \Rightarrow q}%
\quad \left[ q\in \mathrm{VAR}\left( \Gamma ,\gamma \right) ,p\neq q\right] $%
}$

$\fbox{$\left( \text{\textsc{M}}E\!\rightarrow \rightarrow \right) :\ \ \ 
\dfrac{\Gamma ,\alpha ,\beta \rightarrow \gamma \Rightarrow \beta \quad
\Gamma ,\gamma \Rightarrow q}{\Gamma ,\left( \alpha \rightarrow \beta
\right) \rightarrow \gamma \Rightarrow q}\quad \left[ q\in \mathrm{VAR}%
\left( \Gamma ,\gamma \right) \right] $}$

\begin{claim}
\textsc{LM}$_{\rightarrow }$ is sound and complete with respect to minimal
propositional logic \cite{Johansson} and tree-like deducibility. Thus any
given formula $\rho $ is valid in the minimal logic iff it (i.e. sequent $%
\Rightarrow \rho $) is tree-like deducible in \textsc{LM}$_{\rightarrow }$.
\end{claim}

\begin{proof}
Easily follows from \cite{Hudelmaier}.
\end{proof}

Recall that for any (tree-like or dag-like) deduction $\partial $ we denote
by $h\left( \partial \right) $ and $\phi \left( \partial \right) $ its
height and foundation, respectively. Furthermore for any sequent (in
particular, formula) $S$ we denote by $\left| S\right| $ the total number of
`$\rightarrow $'-occurrences in $S$ and following \cite{Hudelmaier} define
the complexity degree $\deg \left( S\right) $:

\begin{enumerate}
\item  $\deg \left( \Gamma ,\alpha \rightarrow \beta \Rightarrow \alpha
\right) :=\left| \alpha \rightarrow \beta \right| +\underset{\xi \in \Gamma 
}{\sum }\left| \xi \right| ,$

\item  $\deg \left( \Gamma \Rightarrow \alpha \right) :=\left| \alpha
\right| +\underset{\xi \in \Gamma }{\sum }\left| \xi \right| ,\ if\ \left(
\nexists \beta \right) :\alpha \rightarrow \beta \in \Gamma .$
\end{enumerate}

\begin{lemma}
\qquad

\begin{enumerate}
\item  Tree-like \textsc{LM}$_{\rightarrow }$ deductions share the
semi-subformula property, where semi-subformulas of $\left( \alpha
\rightarrow \beta \right) \rightarrow \gamma $ include $\beta \rightarrow
\gamma $ along with proper subformulas $\alpha \rightarrow \beta $, $\alpha $%
, $\beta $, $\gamma $. In particular, any $\alpha $ occurring in a \textsc{LM%
}$_{\rightarrow }$ deduction $\partial $ of $\Rightarrow \rho $ is a
semi-subformula of $\rho $, and hence $\left| \alpha \right| \leq \left|
\rho \right| $. Thus $\mu \left( \partial \right) \leq \left| \rho \right| $.

\item  If $S^{\prime }$ occurs strictly above $S$ in a given tree-like 
\textsc{LM}$_{\rightarrow }$ deduction $\partial $, then $\deg \left(
S^{\prime }\right) <\deg \left( S\right) $.

\item  The height of any tree-like \textsc{LM}$_{\rightarrow }$ deduction $%
\partial $ of $S$ is linear in $\left| S\right| $. In particular if $S$ is $%
\Rightarrow \rho $, then $h\left( \partial \right) \leq 3\left| \rho \right| 
$.

\item  The foundation of any tree-like \textsc{LM}$_{\rightarrow }$
deduction $\partial $ of $S$ is at most quadratic in $\left| S\right| $. In
particular if $S$ is $\Rightarrow \rho $, then $\phi \left( \partial \right)
\leq \left( \left| \rho \right| +1\right) ^{2}$.
\end{enumerate}
\end{lemma}

\begin{proof}
1: Obvious. Note that $\beta \rightarrow \gamma $ occurring in premises of $%
\left( \text{\textsc{M}}I2\rightarrow \right) $ and $\left( \text{\textsc{M}}%
E\rightarrow \rightarrow \right) $ are semi-subformulas of $\left( \alpha
\rightarrow \beta \right) \rightarrow \gamma $ occurring in the\ conclusions.

2--3: See \cite{Hudelmaier}.

4: Let $\mathrm{ssf}\left( \alpha \right) $ be the total number of distinct
occurrences of semi-subformulas in a given formula $\alpha $. It is readily
seen that $\mathrm{ssf}\left( -\right) $ satisfies the following three
conditions.

\begin{enumerate}
\item  $\mathrm{ssf}\left( p\right) =1.$

\item  $\mathrm{ssf}\left( p\rightarrow \alpha \right) =2+\mathrm{ssf}\left(
\alpha \right) .$

\item  $\mathrm{ssf}\left( \left( \alpha \rightarrow \beta \right)
\rightarrow \gamma \right) =1+\mathrm{ssf}\left( \alpha \rightarrow \beta
\right) +\mathrm{ssf}\left( \beta \rightarrow \gamma \right) -\mathrm{ssf}%
\left( \beta \right) .$
\end{enumerate}

Moreover 1--3 can be viewed as recursive clauses defining $\mathrm{ssf}%
\left( \alpha \right) $, for any $\alpha $. Having this we easily arrive at $%
\mathrm{ssf}\left( \alpha \right) \leq \left( \left| \alpha \right|
+1\right) ^{2}$ (see Appendix A), which by the assertion 1 yields $\phi
\left( \partial \right) \leq \mathrm{ssf}\left( \rho \right) \leq \left(
\left| \rho \right| +1\right) ^{2}$, as required, provided that $\Rightarrow
\rho $ is the endsequent of $\partial $.
\end{proof}

\subsection{ND calculus \textsc{NM}$_{\rightarrow }$ and embedding of 
\textsc{LM}$_{\rightarrow }$}

Denote by \textsc{NM}$_{\rightarrow }$ a ND proof system for minimal logic
that contains just two rules $\left( \rightarrow I\right) $, $\left(
\rightarrow E\right) $ \cite{Prawitz} (we write `$\rightarrow $' instead of `%
$\supset $'). 
\begin{equation*}
\fbox{$\left( \rightarrow I\right) :\dfrac{\QDATOP{\QDATOP{\left[ \alpha %
\right] }{\vdots }}{\beta }}{\alpha \rightarrow \beta }$}\quad \fbox{$\left(
\rightarrow E\right) :\dfrac{\alpha \quad \alpha \rightarrow \beta }{\beta \ 
}$}
\end{equation*}

\begin{claim}[Prawitz]
\textsc{NM}$_{\rightarrow }$ is sound and complete with respect to minimal
propositional logic and tree-like deducibility.
\end{claim}

\begin{proof}
See \cite{Prawitz}.
\end{proof}

\begin{theorem}
There exists a recursive operator $\mathcal{F}$ that transforms any given
tree-like \textsc{LM}$_{\rightarrow }$ deduction $\partial $ of $\Gamma
\Rightarrow \rho $ into a tree-like \textsc{NM}$_{\rightarrow }$ deduction $%
\mathcal{F}\left( \partial \right) $ with root-formula $\rho $ and
assumptions occurring in $\Gamma $. Moreover $\partial $ and $\mathcal{F}%
\left( \partial \right) $ share the semi-subformula property and linear
(polynomial) upper bounds on the height (resp. foundation). If $\Gamma
=\emptyset $, then $\mathcal{F}\left( \partial \right) $ is a tree-like 
\textsc{NM}$_{\rightarrow }$ proof of $\rho $ such that $\fbox{$h\left( 
\mathcal{F}\left( \partial \right) \right) \leq 18\left| \rho \right| $ and $%
\phi \left( \mathcal{F}\left( \partial \right) \right) <\left( \left| \rho
\right| +1\right) ^{2}\left( \left| \rho \right| +2\right) $ and $\mu \left( 
\mathcal{F}\left( \partial \right) \right) \leq 2\left| \rho \right| $}.$
\end{theorem}

\begin{proof}
$\mathcal{F}\left( \partial \right) $ is defined by straightforward
recursion on $h\left( \partial \right) $ by standard pattern $sequent\
deduction\hookrightarrow $ $\ natural\ deduction$, where sequent deduction
of $\Gamma \Longrightarrow \alpha $ is interpreted as a ND deduction of $%
\alpha $ from open assumptions occurring in $\Gamma $. The recursive clauses
are as follows.\medskip

\begin{enumerate}
\item  
\begin{equation*}
\fbox{$\left( \text{\textsc{M}}A\right) :\Gamma ,p\Rightarrow p$}\overset{%
\mathcal{F}}{\hookrightarrow }\fbox{$p$}\medskip
\end{equation*}

\item  \medskip 
\begin{equation*}
\fbox{$\left. 
\begin{array}{c}
\left( \text{\textsc{M}}I1\rightarrow \right) \ :\dfrac{\Gamma ,\alpha
\Rightarrow \beta }{\Gamma \Rightarrow \alpha \rightarrow \beta } \\ 
\left[ \left( \nexists \gamma \right) :\left( \alpha \rightarrow \beta
\right) \rightarrow \gamma \in \Gamma \right]
\end{array}
\right. $}\overset{\mathcal{F}}{\hookrightarrow }\fbox{$\dfrac{\QDATOP{%
\QDATOP{\left[ \alpha \right] }{\Downarrow }}{\beta }}{\alpha \rightarrow
\beta }\left( \rightarrow I\right) $}
\end{equation*}

\item  \medskip 
\begin{equation*}
\begin{array}{c}
\fbox{$\left( \text{\textsc{M}}I2\rightarrow \right) :\ \dfrac{\Gamma
,\alpha ,\beta \rightarrow \gamma \Rightarrow \beta }{\Gamma ,\left( \alpha
\rightarrow \beta \right) \rightarrow \gamma \Rightarrow \alpha \rightarrow
\beta }$}\overset{\mathcal{F}}{\hookrightarrow }\medskip \\ 
\fbox{$\QDATOP{\QDATOP{\left[ \alpha \right] ^{1}}{\Downarrow }\!\!\!\QDATOP{%
\QDATOP{\dfrac{\dfrac{\left[ \beta \right] ^{2}}{\alpha \rightarrow \beta }%
\left( \rightarrow I\right) \quad \QDATOP{{}}{\left( \alpha \rightarrow
\beta \right) \rightarrow \gamma }}{\underset{\_\_\_\_\_\_\_\_\_\_\_\_\_}{\
\ \gamma }}\left( \rightarrow E\right) }{\beta \rightarrow \gamma ^{\,\left[
2\right] }\ \ \quad }\left( \rightarrow I\right) }{\Downarrow \ \quad \quad
\ \quad \quad }}{\QDATOP{\!\!\!\!\searrow \!\!\!\!\!\searrow \quad \ \ \
\quad \ \ \quad \swarrow \!\!\!\!\!\swarrow \quad \quad \quad \quad \quad }{%
\dfrac{\!\!\!\!\beta \ }{\alpha \rightarrow \beta ^{\,\left[ 1\right] }}%
\left( \rightarrow I\right) \quad \quad }\quad \qquad \qquad }$}
\end{array}
\end{equation*}

\item  \medskip 
\begin{equation*}
\begin{array}{c}
\fbox{$\left. 
\begin{array}{c}
\left( \text{\textsc{M}}E\rightarrow P\right) :\dfrac{\Gamma ,p,\gamma
\Rightarrow q}{\Gamma ,p,p\rightarrow \gamma \Rightarrow q} \\ 
\left[ q\in \mathrm{VAR}\left( \Gamma ,\gamma \right) ,p\neq q\right]
\end{array}
\right. $}\overset{\mathcal{F}}{\hookrightarrow }\medskip \\ 
\fbox{$\QDATOP{\QDATOP{\quad \quad \quad }{\QDATOP{\quad \quad \QDATOP{p}{%
\Downarrow }\ \ \QDATOP{\dfrac{p\quad p\rightarrow \gamma }{\gamma \ }\left(
\rightarrow E\right) }{\Downarrow \ \ \quad \quad \ \ }\quad }{\searrow
\!\!\!\!\!\searrow \ \ \ \swarrow \!\!\!\!\!\swarrow \quad \quad \quad \quad 
}}}{q\quad \quad \quad \quad }$}
\end{array}
\end{equation*}

\item  \medskip\ 
\begin{equation*}
\begin{array}{c}
\fbox{$\left( \text{\textsc{M}}E\rightarrow \rightarrow \right) :\dfrac{%
\Gamma ,\alpha ,\beta \rightarrow \gamma \Rightarrow \beta \qquad \Gamma
,\gamma \Rightarrow q}{\Gamma ,\left( \alpha \rightarrow \beta \right)
\rightarrow \gamma \Rightarrow q}\ \left[ q\in \mathrm{VAR}\left( \Gamma
,\gamma \right) \right] $}\overset{\mathcal{F}}{\hookrightarrow }\medskip \\ 
\fbox{$\dfrac{\dfrac{\QDATOP{\QDATOP{\QDATOP{\left[ \alpha \right] ^{1}}{%
\Downarrow }\ \QDATOP{\QDATOP{\dfrac{\dfrac{\left[ \beta \right] ^{2}}{%
\alpha \rightarrow \beta }\quad \QDATOP{{}}{\left( \alpha \rightarrow \beta
\right) \rightarrow \gamma }}{\underset{\_\_\_\_\_\_\_\_\_\_\_}{\gamma }}}{%
\beta \rightarrow \gamma ^{\,\left[ 2\right] }}}{\Downarrow \ \ }}{\QDATOP{%
\searrow \!\!\!\!\!\searrow \ \ \ \quad \ \ \quad \swarrow
\!\!\!\!\!\swarrow }{\underset{\_\_\_\_\_\_\_\_\_\_\_\_}{\beta }\ }\quad
\qquad \qquad }}{\alpha \rightarrow \beta ^{\,\left[ 1\right] }\quad \qquad
\quad \quad }\QDATOP{\QDATOP{\QDATOP{\QDATOP{\QDATOP{\QDATOP{{}}{{}}}{{}}}{{}%
}}{{}}}{\QDATOP{{}}{{}}}}{\left( \alpha \rightarrow \beta \right)
\rightarrow \gamma }}{\gamma }\ \dfrac{\QDATOP{\QDATOP{\left[ \gamma \right]
_{3}}{\Downarrow }}{q}}{\gamma \rightarrow q^{\,\left[ 3\right] }}}{\quad
\quad \quad \quad \quad \quad \quad q}$}
\end{array}
\end{equation*}
\medskip
\end{enumerate}

Note that each embedding clause increases the height at most by $6$ (just as
in the case $\left( \text{\textsc{M}}E\rightarrow \rightarrow \right) $),
which yields $h\left( \mathcal{F}\left( \partial \right) \right) \leq 6\cdot
h\left( \partial \right) \leq 18\left| \rho \right| $ according to Lemma 2
(3). By the same token, formulas occurring in $\mathcal{F}\left( \partial
\right) $ include the ones occurring in $\partial $ together with possibly
new formulas $\gamma \rightarrow q$ (with old $\gamma $ and $q$) shown on
the right-hand side in the case $\left( \text{\textsc{M}}E\rightarrow
\rightarrow \right) $. There are at most $\phi \left( \partial \right) $ and 
$\left| \rho \right| +1$\ such $\gamma $ and $q$, respectively. Hence by
Lemma 2 (1, 4) we arrive at $\phi \left( \mathcal{F}\left( \partial \right)
\right) <\left( \left| \rho \right| +1\right) ^{2}+\left( \left| \rho
\right| +1\right) ^{2}\left( \left| \rho \right| +1\right) =\left( \left|
\rho \right| +1\right) ^{2}\left( \left| \rho \right| +2\right) $ and $\mu
\left( \mathcal{F}\left( \partial \right) \right) \leq 2\left| \rho \right| $%
, as required.
\end{proof}

\subsection{Horizontal tree-to-dag compression in \textsc{NM}$_{\rightarrow
} $}

We claim that any given tree-like \textsc{NM}$_{\rightarrow }$ deduction $%
\partial $ with root formula $\rho $ can be compressed into a dag-like 
\textsc{NM}$_{\rightarrow }$ deduction $\partial ^{\text{\textsc{c}}}$ of
the same conclusion $\rho $\ such that the size of $\partial ^{\text{\textsc{%
c}}}$ is at most $h\left( \partial \right) \times \phi \left( \partial
\right) $. In particular, if $\partial =\mathcal{F}\left( \partial
_{0}\right) $ for $\partial _{0}$ being a tree-like \textsc{LM}$%
_{\rightarrow }$ deduction of $\Rightarrow \rho $ and $\mathcal{F}$ the
embedding of Theorem 4, then $\partial ^{\text{\textsc{c}}}$ will be a
desired dag-like $\left| \rho \right| $-polysize \textsc{NM}$_{\rightarrow }$
deduction of $\rho $\textsc{. }The operation $\partial \hookrightarrow
\partial ^{\text{\textsc{c}}}$ (that we call \emph{horizontal compression})
runs by bottom-up recursion on $h\left( \partial \right) $ such that for any 
$n\leq h\left( \partial \right) $, the $n^{th}$\ horizontal section of $%
\partial ^{\text{\textsc{c}}}$ is obtained by merging all nodes with
identical formulas occurring in the $n^{th}$\ horizontal section of $%
\partial $ (this operation we call \emph{horizontal collapsing}). Thus the
horizontal compression is obtained by bottom-up iteration of the horizontal
collapsing. $\left| \partial ^{\text{\textsc{c}}}\right| \leq h\left(
\partial \right) \times \phi \left( \partial \right) $ is obvious, as the
size of every (compressed) $n^{th}$\ horizontal section of $\partial ^{\text{%
\textsc{c}}}$ can't exceed $\phi \left( \partial \right) $. It remains to
show that horizontal compression preserves the discharged assumptions. This
requires a more insightful consideration of dag-like deducibility that we
elaborate below.

\subsection{Dag-like deducibility in \textsc{NM}$_{\rightarrow }$}

We wish to elaborate, and work in, the space of dag-like natural deductions.
To begin with we observe that horizontal collapsing may extend the premises
of the underlying inferences. So let us denote by \textsc{NM}$_{\rightarrow
}^{\ast }$ a tree-like extension of \textsc{NM}$_{\rightarrow }$ that
contains multipremise rules of inference of the form

\begin{equation*}
\fbox{$\left( M\right) :\dfrac{\Gamma }{\gamma \ }$}
\end{equation*}
instead of original\ \textsc{NM}$_{\rightarrow }$\ rules $\left( \rightarrow
I\right) $, $\left( \rightarrow E\right) $. Here $\Gamma $ is a multiset
containing $\gamma $, and/or $\beta $, if $\gamma =\alpha \rightarrow \beta $%
, and/or arbitrary $\delta _{i}$ together with $\delta _{i}\rightarrow
\gamma $ $\left( i\in \left[ m\right] \right) $. Thus in particular, $\left(
M\right) $ includes repetition rules 
\begin{equation*}
\fbox{$\left( R\right) :\dfrac{\gamma }{\gamma \ }$}\ \fbox{$\left( R\right)
^{\ast }:\dfrac{\gamma \ \ \cdots \ \ \gamma }{\gamma \ }$}
\end{equation*}
as well as following inferences 
\begin{eqnarray*}
&&\fbox{$\left( \rightarrow I\right) ^{\ast }:\dfrac{\QDATOP{\QDATOP{\left[
\alpha \right] }{\vdots }}{\beta }\ \QDATOP{\QDATOP{{}}{{}}}{\cdots }\QDATOP{%
\QDATOP{\left[ \alpha \right] }{\vdots }}{\beta }}{\alpha \rightarrow \beta }
$}\ \fbox{$\left( \rightarrow I,R\right) ^{\ast }:\dfrac{\QDATOP{\QDATOP{%
\left[ \alpha \right] }{\vdots }}{\beta }\ \QDATOP{\QDATOP{{}}{{}}}{\cdots }%
\QDATOP{\QDATOP{\left[ \alpha \right] }{\vdots }}{\beta }\QDATOP{\QDATOP{{}}{%
{}}}{\gamma }\ \QDATOP{\QDATOP{{}}{{}}}{\cdots }\QDATOP{\QDATOP{{}}{{}}}{%
\gamma }}{\alpha \rightarrow \beta }$} \\
&&\fbox{$\left( \rightarrow E\right) ^{\ast }:\dfrac{\delta _{1}\ \ \delta
_{1}\rightarrow \gamma \ \cdots \ \delta _{m}\ \ \delta _{m}\rightarrow
\gamma }{\gamma \ }$} \\
&&\fbox{$\left( \rightarrow E,R\right) ^{\ast }:\dfrac{\delta _{1}\ \ \delta
_{1}\rightarrow \gamma \ \ \cdots \ \ \delta _{m}\ \ \delta _{m}\rightarrow
\gamma \ \ \gamma \ \cdots \ \gamma }{\gamma \ }$} \\
&&\fbox{$\left( \rightarrow I,E\right) :\dfrac{\QDATOP{\QDATOP{\left[ \alpha %
\right] }{\vdots }}{\beta }\quad \QDATOP{\QDATOP{{}}{{}}}{\delta }\quad 
\QDATOP{\QDATOP{{}}{{}}}{\delta \rightarrow \left( \alpha \rightarrow \beta
\right) }}{\alpha \rightarrow \beta \ }$} \\
&&\fbox{$\left( \rightarrow I,E,R\right) :\dfrac{\QDATOP{\QDATOP{\left[
\alpha \right] }{\vdots }}{\beta }\quad \QDATOP{\QDATOP{{}}{{}}}{\delta }%
\quad \QDATOP{\QDATOP{{}}{{}}}{\delta \rightarrow \left( \alpha \rightarrow
\beta \right) }\quad \QDATOP{\QDATOP{{}}{{}}}{\alpha \rightarrow \beta }}{%
\alpha \rightarrow \beta \ }$}
\end{eqnarray*}
Discharging in \textsc{NM}$_{\rightarrow }^{\ast }$ is inherited from 
\textsc{NM}$_{\rightarrow }$ via sub-occurrences of $\left( \rightarrow
I\right) $.

\begin{lemma}
Tree-like provability in \textsc{NM}$_{\rightarrow }^{\ast }$ is sound and
complete with respect to minimal propositional logic.
\end{lemma}

\begin{proof}
Completeness follows from Claim 3, as \textsc{NM}$_{\rightarrow }$ is
contained in \textsc{NM}$_{\rightarrow }^{\ast }$. Soundness is obvious, as
each $\left( M\right) $ strengthens valid rules $\left( R\right) $, $\left(
\rightarrow I\right) $ and/or $\left( \rightarrow E\right) $.
\end{proof}

Further on we upgrade \textsc{NM}$_{\rightarrow }^{\ast }$ to a desired
dag-like extension, \textsc{NM}$_{\rightarrow }^{\star }$. Let us start with
informal description (cf. formal definitions below). We'll consider only 
\emph{regular dags} (abbr.:\emph{\ redags}), which are specified as rooted
monoedge dags $\partial $ (the roots being the lowest vertices) whose
vertices (also called nodes) admit universal (i.e. path-invariant) height
assignment such that all leaves $x$ have the same height $h\left( x\right)
=h\left( \partial \right) $. We regard \textsc{NM}$_{\rightarrow }^{\star }$
deductions as labeled redags $\partial $ whose nodes can have arbitrary many
children and parents (as usual the roots, $\varrho \left( \partial \right) $%
, have no parents and the leaves have no children). Distinct children are
either singletons or conjugate pairs (mutually separated by fixed partitions 
\textsc{s}). Moreover, all nodes of $\partial $ are labeled with formulas by
a fixed assignment $\ell ^{\text{\textsc{f}}}$. The inferences $\left(
M\right) $ associated with $\partial $ are determined by standard local
correctness conditions on $\ell ^{\text{\textsc{f}}}$ and \textsc{s}, such
that $\ell ^{\text{\textsc{f}}}\left( \varrho \left( \partial \right)
\right) =\rho $, while children's $\ell ^{\text{\textsc{f}}}$-formulas
either coincide with the conclusion's ones or are premises $\beta $ of the
conclusion's $\ell ^{\text{\textsc{f}}}$-formulas $\alpha \rightarrow \beta $%
, or else are conjugate premises $\delta _{i}$, $\delta _{i}\rightarrow
\gamma $ of the conclusion's $\ell ^{\text{\textsc{f}}}$-formulas $\gamma $.
Besides, there is a fixed assignment $\ell ^{\text{\textsc{g}}}$ that is
defined for any edge $e=\left\langle u,v\right\rangle $ ($u$ being a parent
of $v$) that admits inverse branching below $u$. To put it more precisely we
consider descending chains $K\left( u\right) =\left[ u=x_{0},\cdots ,x_{k}%
\right] $ ($k>0$) in $\partial $ such that for all $0<i<k$, $x_{i}$ has
exactly one parent $x_{i+1}$,whereas $x_{k}$ has at least two parents (such $%
K\left( u\right) $ is uniquely determined by $u$). Having this we regard $%
\ell ^{\text{\textsc{g}}}\left( e\right) $ as a chosen nonempty set of
parents of $x_{k}$, called $\ell ^{\text{\textsc{g}}}$\emph{-grandparents}
of $v$ with respect to $u$. It is assumed that $\ell ^{\text{\textsc{g}}%
}\left( e\right) \subseteq \ell ^{\text{\textsc{g}}}\left( \left\langle
x_{i},x_{i-1}\right\rangle \right) $ holds for all $1<i\leq k$, while all
parents of $x_{k}$\ are $\ell ^{\text{\textsc{g}}}$-grandparents of some $%
x_{k-1}$'s children (with respect to $x_{k}$). Descending \emph{deduction
threads} connecting leaves with the root are naturally determined by $\ell ^{%
\text{\textsc{g}}}$-grandparents that are regarded as ``road signs'' showing
allowed ways from the leaves down to the root, when passing from $v$ to $%
x_{k}$ through $u$, as specified above. These parameters determine `global' 
\emph{discharging function} on the set of top formulas (also called
assumptions).

\subsubsection{Formal definitions}

\begin{definition}
Consider a rooted monoedge redag $D=\left\langle \text{\textsc{v}}\left(
D\right) ,\text{\textsc{e}}\left( D\right) \right\rangle $, \textsc{e}$%
\left( D\right) \subset \,$\textsc{v}$\left( D\right) ^{2}$. \textsc{v}$%
\left( D\right) $ and \textsc{e}$\left( D\right) $ are called the \emph{%
vertices} (or \emph{nodes}) and the \emph{edges} (ordered), respectively; if 
$\left\langle u,v\right\rangle \in \,$\textsc{e}$\left( D\right) $, then $u$
and $v$ are called \emph{parents} and \emph{children} of each other,
respectively. For any $u\in \,$\textsc{v}$\left( D\right) $ denote by $%
h\left( u,D\right) \geq 0$\ the \emph{height} of $u$ and let $h\!\left(
D\right) :=\max \left\{ h\left( u,D\right) :u\in \text{\textsc{v}}\left(
D\right) \right\} $ (the \emph{height }of $D$). Any $u\in \,$\textsc{v}$%
\left( D\right) $ has $\overrightarrow{\deg }\left( u,D\right) \geq 0$
children \textsc{c}$\left( u,D\right) :=\!\left\{ u^{\left( 1\right)
},\cdots ,u^{\left( \overrightarrow{\deg }\left( u,D\right) \right)
}\right\} $ and $\overleftarrow{\deg }\left( u,D\right) \geq 0$ parents 
\textsc{p}$\left( u,D\right) :=\!\left\{ u_{\left( 1\right) },\cdots
,u_{\left( \overleftarrow{\deg }\left( u,D\right) \right) }\right\} $ (both
ordered). \footnote{{\footnotesize That is, }$\overrightarrow{\deg }\left(
u,D\right) ${\footnotesize \ (resp. }$\overleftarrow{\deg }\left( u,D\right) 
${\footnotesize ) is the total \ number of targets with source }$u$%
{\footnotesize \ (resp. total number of sources with target }$u$%
{\footnotesize ), in }$D$.} Let \textsc{l}$\left( D\right) \!:=\left\{ u\in 
\text{\textsc{v}}\left( D\right) :\overrightarrow{\deg }\left( u,D\right)
=0\right\} $ (\emph{leaves}), and $\varrho \left( D\right) :=$ the root of $%
D $; thus \textsc{p}$\left( u,D\right) =\emptyset \Leftrightarrow u=\varrho
\!\left( D\right) \Leftrightarrow h\left( u,D\right) =0$ and \textsc{c}$%
\left( u,D\right) =\emptyset \Leftrightarrow u\in \text{\textsc{l}}\left(
D\right) \Leftrightarrow h\left( u,D\right) =h\left( D\right) $. With every $%
u\in \,$\textsc{v}$\left( D\right) \setminus $\textsc{l}$\left( D\right) $
we associate a fixed partition \footnote{{\footnotesize not necessarily
disjoint.}} \textsc{s}$\left( u,D\right) \subset \text{\textsc{c}}\left(
u,D\right) \cup \,$\textsc{c}$\left( u,D\right) ^{2}$ such that \textsc{c}$%
\left( u,D\right) \!=\!\left( \text{\textsc{s\negthinspace }}\left(
u,\!D\right) \cap \text{\textsc{v\negthinspace }}\left( D\right) \right)
\,\!\cup \!\,\left\{ x,y:\!\left\langle x,y\right\rangle \in \!\text{\textsc{%
s\negthinspace }}\left( u,D\right) \right\} $. Set \textsc{s}$\left(
D\right) :=\!\underset{u\in \text{\textsc{v}}\left( D\right) \setminus \text{%
\textsc{l}}\left( D\right) }{\bigcup }$\textsc{s}$\left( u,\!D\right) $, to
be abbreviated by \textsc{s}. By the same token, we'll often drop `$D$' in $%
h\left( u,D\right) $, $\overrightarrow{\deg }\left( u,D\right) $, $%
\overleftarrow{\deg }\left( u,D\right) $, \textsc{c}$\left( u,D\right) $, 
\textsc{p}$\left( u,D\right) $, $\varrho \left( D\right) $, $K\!\left(
u,\!D\right) $, $U\!\left( u,\!D\right) $ (see below), if $D$\ is clear from
the context. We let \thinspace \textsc{e}$_{0}\!\left( D\right) :=\!\left\{
\!\left\langle u,v\right\rangle \in \!\text{\textsc{e}}\left( D\right) :v\in 
\text{\textsc{l}}\left( D\right) \!\right\} $ (\emph{top edges}) and use
abbreviations $x\prec _{D}y:\Leftrightarrow `x$ \emph{occurs strictly below }%
$y$\emph{, in }$D$' and $x\preceq _{D}y:\Leftrightarrow x\preceq _{D}y\vee
x=y$. For any $u\!\in \,$\textsc{v}$\left( D\right) \!$ we let $K\!\left(
u,\!D\right) \!=\!\left[ u\!=\!x_{0},\cdots ,x_{k}=:U\!\left( u,\!D\right) 
\right] $ be the uniquely determined descending chain of maximal length such
that either $\overleftarrow{\deg }\left( u\right) \neq 1$ and $k=0$ or else $%
\left\langle x_{i+1},x_{i}\right\rangle \in \,$\textsc{e}$\left( D\right) $
and $\overleftarrow{\deg }\left( x_{i}\right) =1$, for all $i<k$. If Let 
\textsc{e}$_{\triangle }\!\left( D\right) :=\left\{ e=\left\langle
u,v\right\rangle \in \,\text{\textsc{e}}\left( D\right) :U\!\left(
u,D\right) \neq \varrho \right\} $. Thus $\overleftarrow{\deg }\left(
U\!\left( u\right) \right) >1$ holds for every $e=\left\langle
u,v\right\rangle \in \,$\textsc{e}$_{\triangle }\!\left( D\right) $. Note
that \textsc{e}$_{\triangle }\!\left( D\right) =\emptyset $, if $D$ is a
tree.

Let $\partial =\left\langle D,\text{\textsc{s}},\ell ^{\text{\textsc{f}}%
},\ell ^{\text{\textsc{g}}}\right\rangle $ extend $\left\langle D,\text{%
\textsc{s}}\right\rangle $ by labeling functions $\ell ^{\text{\textsc{f}}}:$%
\textsc{\thinspace \thinspace v}$\left( D\right) \!\rightarrow \ $\textsc{%
\negthinspace f}$\left( \mathcal{L}_{\rightarrow }\right) $\ and $\ell ^{%
\text{\textsc{g}}}:$\textsc{\thinspace \thinspace e}$_{\triangle }\!\left(
D\right) \rightarrow \!\,\wp \left( \text{\textsc{v\negthinspace }}\left(
D\right) \right) $, where \textsc{\negthinspace f}$\left( \mathcal{L}%
_{\rightarrow }\right) $ is the set of $\mathcal{L}_{\rightarrow }$
formulas. $\partial $ is called a \emph{plain\ }(or \emph{unencoded}) \emph{%
dag-like }\textsc{NM}$_{\rightarrow }^{\star }\!$\emph{\ deduction} iff the
following local correctness conditions hold (along with standard ones with
regard to $\left\langle D,\text{\textsc{s}}\right\rangle $).

\begin{enumerate}
\item  For any $u\in \,$\textsc{v}$\left( D\right) $\ and $x,y\in \,$\textsc{%
c}$\left( u\right) $ it holds:

\begin{enumerate}
\item  $h\left( x\right) =h\left( y\right) =h\left( u\right) +1$,

\item  if $x\in \,$\textsc{s}$\left( u\right) $ then either $\ell ^{\text{%
\textsc{f}}}\left( u\right) =\ell ^{\text{\textsc{f}}}\left( x\right) $

or $\ell ^{\text{\textsc{f}}}\left( u\right) =\alpha \rightarrow \ell ^{%
\text{\textsc{f}}}\left( x\right) $ [abbr.: $\left\langle u,x\right\rangle
\in \left( \rightarrow I\right) _{\alpha }$]

for a (uniquely determined) $\alpha \in \,$\textsc{f}$\left( \mathcal{L}%
_{\rightarrow }\right) $,

\item  $\left\langle x,y\right\rangle \in \,$\textsc{s}$\left( u\right) $
implies $\ell ^{\text{\textsc{f}}}\left( y\right) =\ell ^{\text{\textsc{f}}%
}\left( x\right) \rightarrow \ell ^{\text{\textsc{f}}}\left( u\right) $.
\end{enumerate}

\item  For any $e\!=\!\left\langle u,v\right\rangle \!\in \!\,\text{\textsc{e%
}}_{\triangle }\!\left( D\right) \!$ and $w\in \,$\textsc{c}$\left( u\right) 
$ it holds:

\begin{enumerate}
\item  $\emptyset \neq \ell ^{\text{\textsc{g}}}\!\left( e\right)
\!\subseteq \,$\textsc{p}$\left( U\!\left( u\right) \right) $,

\item  $\left\langle v,w\right\rangle \in \,$\textsc{s}$\left( u\right) $
implies $\ell ^{\text{\textsc{g}}}\!\left( e\right) =\ell ^{\text{\textsc{g}}%
}\!\left( \left\langle u,w\right\rangle \right) $,

\item  $\overleftarrow{\deg }\left( v\right) =1$ implies $\ell ^{\text{%
\textsc{g}}}\!\left( e\right) \!=\underset{z\in \text{\textsc{c}}\left(
v\right) }{\bigcup }\ell ^{\text{\textsc{g}}}\!\left( \left\langle
v,z\right\rangle \right) $.
\end{enumerate}

\item  For any $u\in \,$\textsc{v}$\left( D\right) \setminus $\textsc{l}$%
\left( D\right) $,

$\overleftarrow{\deg }\left( u\right) >1$ implies \textsc{p}$\left( u\right)
\subseteq \underset{v\in \text{\textsc{c}}\left( u\right) }{\bigcup }\ell ^{%
\text{\textsc{g}}}\!\left( \left\langle u,v\right\rangle \right) $.
\end{enumerate}

Denote by $\mathcal{D}^{\star }$ the set of plain dag-like \textsc{NM}$%
_{\rightarrow }^{\star }$ deductions.
\end{definition}

\begin{definition}
For any $\partial =\left\langle D,\text{\textsc{s}},\ell ^{\text{\textsc{f}}%
},\ell ^{\text{\textsc{g}}}\right\rangle \in \mathcal{D}^{\star }$, $%
e=\!\left\langle u,v\right\rangle \in \!\ $\textsc{e}$\left( D\right) $, $%
z\prec _{D}u$, let \textsc{th}$\left( e,z,\partial \right) $ be the set of 
\emph{deduction threads} $\Theta =\left[ v=x_{0},u=x_{1},\cdots ,x_{n}=z%
\right] $ connecting $e$ with $z$, where any $\Theta $ in question is a
descending chain such that for every $i<n$, $\left\langle
x_{i+1},x_{i}\right\rangle $ $\in \!\ $\textsc{e}$\left( D\right) $ and
either $\overleftarrow{\deg }\left( x_{i}\right) =1$ or else $\overleftarrow{%
\deg }\left( x_{i}\right) >1$ and $x_{i+1}\in \ell ^{\text{\textsc{g}}%
}\!\left( \left\langle x_{j+1},x_{j}\right\rangle \right) $, where $j:=\max
\left\{ k<i:k=0\vee \overleftarrow{\deg }\left( x_{k}\right) >1\right\} $.
Now $\alpha \in \text{\textsc{\negthinspace f}}\left( \mathcal{L}%
_{\rightarrow }\right) $\ is called an \emph{open} (or \emph{undischarged}) 
\emph{assumption} in $\partial $\ if there is a $\Theta \in \!\text{\textsc{%
th\negthinspace }}\left( e\!,\varrho ,\partial \right) $ \negthinspace for $%
e=\!\left\langle u,v\right\rangle \in \!\ $\textsc{e}$_{0}\left( D\right) $
and $\ell ^{\text{\textsc{f}}}\left( v\right) \!=\!\alpha $ $\!$that
\negthinspace contains \negthinspace no $\!\left\langle
x_{i+1},x_{i}\right\rangle \!\in \!\left( \rightarrow \!I\right) _{\alpha }$%
, $i<n$; such $\Theta $\ is called an \emph{open thread}, in $\partial $.
Denote by $\Gamma _{\partial }$\ the set of open assumptions in $\partial $.
Call $\partial $\ a \emph{dag-like }\textsc{NM}$_{\rightarrow }^{\star }$ 
\emph{deduction} of $\rho :=\ell ^{\text{\textsc{f}}}\left( \varrho \right) $%
\ from $\Gamma _{\partial }$. If $\Gamma _{\partial }=\emptyset $, then is
called a \emph{dag-like }\textsc{NM}$_{\rightarrow }^{\star }$ \emph{proof}
of $\rho $.
\end{definition}

In the sequel \textsc{NM}$_{\rightarrow }^{\star }$ deductions (proofs) are
also called \emph{plain dag-like} \textsc{NM}$_{\rightarrow }$ \emph{%
deductions} (\emph{proofs}). \footnote{{\footnotesize Here and below `plain'
means `unencoded' (see 3.6.1, below)}}\ Note that in the tree-like domain
such dag-like (actually redag-like) provability is equivalent to canonical
tree-like \textsc{NM}$_{\rightarrow }$ provability. Indeed, in any tree-like
deduction, every leaf has exactly one deduction thread, and hence $\ell ^{%
\text{\textsc{g}}}$ can be dropped entirely. Also note that \textsc{NM}$%
_{\rightarrow }^{\ast }$ (and hence also \textsc{NM}$_{\rightarrow }$) is
tree-like embeddable into \textsc{NM}$_{\rightarrow }^{\star }$ by iterating
the repetition rule $\left( R\right) $, if necessary, in order to fulfill
the redag height condition $h\left( x\right) =h\left( \partial \right) $,
for all leaves $x$. Obviously this operation preserves $h\left( \partial
\right) $, $\phi \left( \partial \right) $ and $\mu \left( \partial \right) $%
.

\subsection{Horizontal compression continued}

Let us go back to the horizontal compression $\partial \hookrightarrow
\partial ^{\text{\textsc{c}}}$, where without loss of generality we assume
that $\partial $ is an arbitrary tree-like \textsc{NM}$_{\rightarrow
}^{\star }$ deduction of $\rho $. \footnote{{\footnotesize That is, every
node }$x\neq \varrho \left( \partial \right) ${\footnotesize \ has exactly
one parent.}}\ To complete our recursive definition of $\partial ^{\text{%
\textsc{c}}}$\ via horizontal collapsing (see 3.3 above) it remains to
specify $\ell ^{\text{\textsc{g}}}$. So let us take a closer look at the
structure of $\partial ^{\text{\textsc{c}}}$. For any $n\leq h\left(
\partial \right) $, denote by $\partial _{n}^{\text{\textsc{c}}%
}=\left\langle D_{n},\text{\textsc{s}}_{n},\ell _{n}^{\text{\textsc{f}}%
},\ell _{n}^{\text{\textsc{g}}}\right\rangle $ a deduction that is obtained
after executing the $n^{th}$ recursive step in question. Note that $\partial
_{0}^{\text{\textsc{c}}}=\partial $ and $\partial _{h\left( \partial \right)
}^{\text{\textsc{c}}}=\partial ^{\text{\textsc{c}}}$. Moreover, for any $%
i\leq n<j$ we have $L_{i}\left( D_{n}\right) =L_{i}\left( D_{h\left(
\partial \right) }\right) $, $L_{j}\left( D_{n}\right) =L_{j}\left(
D_{0}\right) $ and $h\left( D_{n}\right) =h\left( D_{0}\right) =h\left(
D_{h\left( \partial \right) }\right) $, where $L_{k}\left( D_{m}\right)
:=\left\{ x\in \text{\textsc{v}}\left( D_{m}\right) :h\left( x\right)
=k\right\} $ (= the $k^{th}$ section of $\partial _{m}^{\text{\textsc{c}}}$%
). Besides, if $n<h\left( \partial \right) $, then all $x\in $ $%
L_{n+1}\left( D_{n}\right) $ are the roots of the corresponding (maximal)
tree-like subgraphs of $\partial $, while $\partial _{n+1}^{\text{\textsc{c}}%
}$ arises from $\partial _{n}^{\text{\textsc{c}}}$\ by merging distinct $%
x\in $ $L_{n+1}\left( D_{n}\right) $ labeled with identical formulas, $\ell
^{\text{\textsc{f}}}\left( x\right) $, and defining edges by the
corresponding homomorphism. Thus $L_{n+1}\left( D_{n+1}\right) \subseteq
L_{n+1}\left( D_{n}\right) $, while $x\neq y\in $ $L_{n+1}\left(
D_{n+1}\right) $ implies $\ell ^{\text{\textsc{f}}}\left( x\right) \neq \ell
^{\text{\textsc{f}}}\left( y\right) $. (If $L_{n+1}\left( D_{n+1}\right)
=L_{n+1}\left( D_{n}\right) $, then $\partial _{n+1}^{\text{\textsc{c}}%
}=\partial _{n}^{\text{\textsc{c}}}$ and $\ell _{n+1}^{\text{\textsc{g}}%
}=\ell _{n}^{\text{\textsc{g}}}$.) Now suppose $L_{n+1}\left( D_{n+1}\right)
\neq L_{n+1}\left( D_{n}\right) $, $n<h\left( D_{0}\right) $, and let $%
M_{n+1}\subseteq L_{n+1}\left( D_{n+1}\right) $ be the set of all merge
points in $\partial _{n+1}^{\text{\textsc{c}}}$. The $\ell _{n+1}^{\text{%
\textsc{g}}}$-grandparents are defined as follows. For any $e=\left\langle
u,v\right\rangle \in $\thinspace \textsc{e}$_{\triangle }\!\left(
D_{n+1}\right) $, $u\in L_{j}\left( D_{n+1}\right) $, $v\in L_{j+1}\left(
D_{n+1}\right) $, $j<h\left( D_{0}\right) $, consider $K\left(
u,D_{n+1}\right) =\left[ u=x_{0},\cdots ,x_{k}\right] $. (Note that $%
x_{i}\in L_{j}\left( D_{n}\right) $ for all but at most one $x_{i}$, $i\leq
k $.) We let $\ell _{n+1}^{\text{\textsc{g}}}\left( e\right) :=\ell _{n}^{%
\text{\textsc{g}}}\left( e\right) $ except for the following two cases.

\begin{enumerate}
\item  Suppose $j=n+1$ and $v\in M_{n+1}$. We let $\ell _{n+1}^{\text{%
\textsc{g}}}\left( e\right) $\ be the union of all $\ell _{n}^{\text{\textsc{%
g}}}\left( \left\langle u,w\right\rangle \right) $ such that $w\in \text{%
\textsc{c}}\left( u,\!D_{n}\right) $ and $\ell ^{\text{\textsc{f}}}\left(
v\right) =\ell ^{\text{\textsc{f}}}\left( w\right) $.

\item  Suppose $h\left( \partial \right) -1\geq j>n+1$, $x_{k}\in M_{n+1}$
and \textsc{p}$\left( x_{k-1},\!D_{n}\right) =\left\{ y\right\} $, while 
\textsc{p}$\left( y,\!D_{n}\right) =\left\{ y_{\left( 1\right) }\right\} $,
i.e. $y_{\left( 1\right) }$ is the only parent of $y$ in $\partial _{n}^{%
\text{\textsc{c}}}$. Then we let $\ell _{n+1}^{\text{\textsc{g}}}\left(
e\right) :=\left\{ y_{\left( 1\right) }\right\} $.
\end{enumerate}

Having this we observe that $\partial _{n+1}^{\text{\textsc{c}}}$ preserves
the open (resp. closed) assumptions of $\partial _{n}^{\text{\textsc{c}}}$.
The same conclusion with regard to $\partial $ and $\partial ^{\text{\textsc{%
c}}}$ follows immediately by induction on $n\leq h\left( D_{0}\right)
=h\left( \partial \right) $. In particular, if $\partial $ is a tree-like 
\textsc{NM}$_{\rightarrow }$ proof of $\rho $, then $\partial ^{\text{%
\textsc{c}}}$ is a plain dag-like \textsc{NM}$_{\rightarrow }$ proof of $%
\rho $. This completes our informal description of the required tree-to-dag
horizontal compression\ $\partial \hookrightarrow \partial ^{\text{\textsc{c}%
}}$. Formal definitions are shown below.

\subsubsection{Horizontal collapsing}

Recall that horizontal compression $\partial \hookrightarrow \partial ^{%
\text{\textsc{c}}}$ is obtained by bottom-up iteration of the \emph{%
horizontal collapsing} that merges distinct nodes labeled with identical
formulas occurring in the same horizontal section of $\partial $. Our next
definition will formalize the latter operation. In the sequel for any $D$
and $x\in \,$\textsc{v}$\left( D\right) $ we let $\left( D\right)
_{x}:=\left\langle \text{\textsc{v}}\left( \text{\negthinspace }\left(
D\right) _{x}\text{\negthinspace }\right) ,\text{\textsc{e}}\left( \left(
D\right) _{x}\right) \right\rangle $ for \textsc{v}$\left( \left( D\right)
_{x}\right) =\left\{ y\in \text{\textsc{v}}\left( D\right) :x\preceq
_{D}y\right\} $ and \textsc{e}$\left( \left( D\right) _{x}\right) =\,$%
\textsc{e}$\left( D\right) \,\cap \,\text{\textsc{v\negthinspace }}\left(
\left( D\right) _{x}\right) ^{2}$. For any $n>0$ we let $L_{n}\left(
D\right) :=\left\{ x\in \text{\textsc{v}}\left( D\right) :h\left( x\right)
=n\right\} $ and denote by $\mathcal{D}_{n}^{\star }$\ the set of dag-like
deductions $\partial =\left\langle D,\text{\textsc{s}},\ell ^{\text{\textsc{f%
}}},\ell ^{\text{\textsc{g}}}\right\rangle \in \mathcal{D}^{\star }$ such
that $\left( D\right) _{x}$ are pairwise disjoint (sub)trees, for all $x\in
L_{n}\left( D_{n}\right) $. Note that $\mathcal{D}_{n}^{\star }=\mathcal{D}%
^{\star }$ for $n>h\!\left( D\right) $, while $\mathcal{D}_{1}^{\star }$
consists of all tree-like \textsc{NM}$_{\rightarrow }^{\ast }$ deductions
(see above). So in the sequel we'll rename $\mathcal{D}_{1}^{\star }$ to $%
\mathcal{T}^{\ast }$ and denote its elements by $\left\langle T,\text{%
\textsc{s}},\ell ^{\text{\textsc{f}}}\right\rangle $, rather than $%
\left\langle D,\text{\textsc{s}},\ell ^{\text{\textsc{f}}},\ell ^{\text{%
\textsc{g}}}\right\rangle $ (recall that $\ell ^{\text{\textsc{g}}}$ is
irrelevant in the tree-like case).

\begin{definition}[horizontal collapsing]
Suppose $\partial =\left\langle D,\text{\textsc{s}},\ell ^{\text{\textsc{f}}%
},\ell ^{\text{\textsc{g}}}\text{\textsc{\thinspace }}\right\rangle \in 
\mathcal{D}_{n}^{\star }$, $n\leq h\!\left( D\right) $, $\alpha \in \,$%
\textsc{f}$\left( \mathcal{L}_{\rightarrow }\right) $ and $S_{n,\alpha
}=\left\{ y\in L_{n}\left( D\right) :\ell ^{\text{\textsc{f}}}\left(
y\right) =\alpha \right\} $, $\left| S_{n,\alpha }\right| >1$. Moreover let $%
r\in S_{n,\alpha }$ be fixed. Let $C_{\alpha }=\underset{y\in S_{n,\alpha }}{%
\bigcup }$\textsc{c}$\left( y,D\right) $ and denote by $\left( D\right)
_{\alpha ,r}$ a tree extending upper subtrees $\underset{z\in C_{\alpha }}{%
\bigcup }\left( D\right) _{z}$ by a new root $r$. We construct a dag-like
deduction $\partial _{n,\alpha }^{\text{\textsc{c}}}=\left\langle
D_{n,\alpha },\text{\textsc{s}}_{n,\alpha },\ell _{n,\alpha }^{\text{\textsc{%
f}}},\ell _{n,\alpha }^{\text{\textsc{g}}}\right\rangle $ by collapsing $%
S_{n,\alpha }$ to $\left\{ r\right\} $. To put it more precisely, we
stipulate:

\begin{enumerate}
\item  $D_{n,\alpha }$\ arises from $D$ by substituting $\left( D\right)
_{\alpha ,r}$ for $\left( D\right) _{r}$ and deleting $\left( D\right) _{y}$
for all $r\neq y\in S_{n,\alpha }$. That is, in the formal terms, we have 
\begin{equation*}
\text{\textsc{v}}\left( D_{n,\alpha }\right) =\left( \text{\textsc{v}}\left(
D\right) \setminus \underset{y\in S_{n,\alpha }}{\bigcup }\text{\textsc{v}}%
\left( \left( D\right) _{y}\right) \right) \cup \,\text{\textsc{v}}\left(
\left( D\right) _{\alpha ,r}\right) \ and\ \text{\textsc{e\thinspace }}%
\left( D_{n,\alpha }\right) \!=
\end{equation*}
\begin{equation*}
\,\left( \!\text{\textsc{e\negthinspace \thinspace }}\left( D\right) \cap 
\text{\textsc{v\negthinspace }}\left( D_{n,\alpha }\right) ^{2}\right) \cup
\left\{ \!\!\left\langle r,v\right\rangle \!:v\!\in \!\underset{y\in
S_{n,\alpha }}{\bigcup }\!\!\text{\textsc{c}}\left( y,\!D\right) \!\right\}
\cup \left\{ \!\!\left\langle u,r\right\rangle \!:u\!\in \!\underset{y\in
S_{n,\alpha }}{\bigcup }\!\!\text{\textsc{p}}\left( y,\!D\right) \!\right\}
\!.
\end{equation*}

\item  For any $u\in \,$\textsc{v}$\left( D_{n,\alpha }\right) $ we define 
\textsc{s}$_{n,\alpha }\!\left( u,D_{n,\alpha }\right) $ by cases as follows.

\begin{enumerate}
\item  If $u\notin \left\{ r\right\} \cup \underset{y\in S_{n,\alpha }}{%
\bigcup }\!$\textsc{p}$\left( y,D\right) $, then \textsc{s}$_{n,\alpha
}\!\left( u,D_{n,\alpha }\right) :=\,$\textsc{s}$\left( u,D\right) $.

\item  \textsc{s}$_{n,\alpha }\!\left( u,D_{n,\alpha }\right) :=\underset{%
y\in S_{n,\alpha }}{\bigcup }$\textsc{s}$\left( y,D\right) $.

\item  Suppose $u\in \underset{y\in S_{n,\alpha }}{\bigcup }\!$\textsc{p}$%
\left( y,D\right) $. We let \textsc{s}$_{n,\alpha }\!\left( u,D_{n,\alpha
}\right) :=X\cup Y$, where

$X=\ \left( \text{\textsc{s\negthinspace }}\left( u,D\right) \cap
L_{n}\!\left( D_{n,\alpha }\right) \right) \cup \left\{ r\right\} $ and

$Y=\left\{ \!\left\langle y_{0},y_{1}\right\rangle \in L_{n}\!\left(
D_{n,\alpha }\right) ^{2}:\!\!\left. \! 
\begin{array}{c}
\left( \exists \left\langle x_{0},x_{1}\right\rangle \in \text{\textsc{s}}%
\left( u,D\right) \right) \left( \forall j\leq 1\right) \\ 
\left( x_{j}=y_{j}\vee \left( r\neq x_{j}\in S_{n,\alpha }\wedge
y_{j}=r\right) \right)
\end{array}
\!\!\right. \!\!\right\} $.
\end{enumerate}

\item  For any $u\in \,$\textsc{v}$\left( D_{n,\alpha }\right) $ we let $%
\ell _{n,\alpha }^{\text{\textsc{f}}}\!\left( u\right) :=\ell ^{\text{%
\textsc{f}}}\left( u\right) $.

\item  For any $e=\!\left\langle u,v\right\rangle \in \text{\textsc{e}}%
_{\triangle }\!\left( D_{n,\alpha }\right) $ \negthinspace and $K\left(
u,D_{n,\alpha }\right) =\left[ u=x_{0},\cdots ,x_{k}\right] $ we define $%
\ell _{n,\alpha }^{\text{\textsc{g}}}\!\left( e\right) $, where $u\in
L_{j}\!\left( D_{n,\alpha }\right) $, $v\in L_{j+1}\!\left( D_{n,\alpha
}\right) $ for $j<h\left( D\right) $. We can just as well assume that $v\in
\,$\textsc{l}$\left( D_{n,\alpha }\right) $ or $\overleftarrow{\deg }\left(
v,D_{n,\alpha }\right) >1$ and define the rest according to clause 2 (c) of
Definition 6 by induction on $h\left( D\right) -j$. So assuming $v\in \,$%
\textsc{l}$\left( D_{n,\alpha }\right) \vee \overleftarrow{\deg }\left(
v,D_{n,\alpha }\right) >1$ consider the following cases. (Note that (c) and
(e)$_{ii}$ are the only cases with $\ell _{n,\alpha }^{\text{\textsc{g}}%
}\!\left( e\right) \neq \ell ^{\text{\textsc{g}}}\!\left( e\right) $.)

\begin{enumerate}
\item  Suppose $j+1<n$. Then $\ell _{n,\alpha }^{\text{\textsc{g}}}\!\left(
e\right) :=\ell ^{\text{\textsc{g}}}\!\left( e\right) $.

\item  Suppose $j+1=n$ and $v\neq r$. Then $\ell _{n,\alpha }^{\text{\textsc{%
g}}}\!\left( e\right) :=\ell ^{\text{\textsc{g}}}\!\left( e\right) $.

\item  Suppose $j+1=n$ and $v=r$. Then $\ell _{n,\alpha }^{\text{\textsc{g}}%
}\!\left( e\right) :=\underset{w\in \text{\textsc{c}}\left( u,D\right) \cap
S_{n,\alpha }}{\bigcup }\ell ^{\text{\textsc{g}}}\left( \left\langle
u,w\right\rangle \right) $.

\item  Suppose $j+1>n$ (and hence $v\in \,$\textsc{l}$\left( D\right) $) and 
$r\notin K\left( u,D_{n,\alpha }\right) $. Then $\ell _{n,\alpha }^{\text{%
\textsc{g}}}\!\left( e\right) :=\ell ^{\text{\textsc{g}}}\!\left( e\right) $.

\item  Suppose $j+1>n$, $x_{k}=r$ and \textsc{p}$\left( x_{k-1},\!D\right)
=\left\{ y\right\} $. Then:

\begin{enumerate}
\item  if $\overleftarrow{\deg }\left( y,D\right) >1$, then $\ell _{n,\alpha
}^{\text{\textsc{g}}}\!\left( e\right) :=\ell ^{\text{\textsc{g}}}\!\left(
e\right) $,

\item  if \textsc{p}$\left( y,\!D\right) =\left\{ y_{\left( 1\right)
}\right\} $ (thus $\overleftarrow{\deg }\left( y,D\right) =1$), then $\ell
_{n,\alpha }^{\text{\textsc{g}}}\!\left( e\right) :=\left\{ y_{\left(
1\right) }\right\} $.
\end{enumerate}
\end{enumerate}
\end{enumerate}

To complete the $\left( n,\alpha \right) $\emph{-collapsing operation} $%
\partial \hookrightarrow \partial _{n,\alpha }^{\text{\textsc{c}}}$, let $%
\partial _{n,\alpha }^{\text{\textsc{c}}}:=\partial $ in the case $\left|
S_{n,\alpha }\right| =1$. Now let $\partial _{n}^{\text{\textsc{c}}}$ arise
from $\partial $ by applying $\left( n,\alpha \right) $-collapsing
successively to all $\alpha =\ell _{n}^{\text{\textsc{f}}}\left( x\right) $, 
$x\in L_{n}\!\left( D\right) $, and arbitrary $r\in S_{n,\alpha }$. Thus $%
\partial _{n}^{\text{\textsc{c}}}$ is the iteration of $\partial _{n,\alpha
}^{\text{\textsc{c}}}$\ with respect to all $\alpha $ occurring in the $%
n^{th}$\ section of $D$. The operation $\partial \hookrightarrow \partial
_{n}^{\text{\textsc{c}}}$ is called the \emph{horizontal collapsing on level}
$n$, in \textsc{NM}$_{\rightarrow }^{\star }$.
\end{definition}

\begin{lemma}
For any $\partial =\left\langle D,\text{\textsc{s}},\ell ^{\text{\textsc{f}}%
},\ell ^{\text{\textsc{g}}}\text{\textsc{\thinspace }}\right\rangle \in 
\mathcal{D}_{n}^{\star }$, $n\leq h\!\left( D\right) $, and $\partial _{n}^{%
\text{\textsc{c}}}=\left\langle D_{n},\text{\textsc{s}}_{n},\ell _{n}^{\text{%
\textsc{f}}},\ell _{n}^{\text{\textsc{g}}}\text{\textsc{\thinspace }}%
\right\rangle $, the following conditions 1--5 hold.

\begin{enumerate}
\item  $\partial _{n}^{\text{\textsc{c}}}\in \mathcal{D}_{n}^{\star }$.

\item  \textsc{v}$\left( D_{n}\right) \subseteq \text{\textsc{v}}\left(
D\right) $, $\varrho \!\left( D_{n}\right) =\varrho \!\left( D\right) $ and $%
h\!\left( D_{n}\right) =h\!\left( D\right) $.

\item  For any $n\neq i\leq h\!\left( D\right) $, $L_{i}\!\left(
D_{n}\right) =L_{i}\!\left( D\right) $, while $L_{n}\!\left( D_{n}\right)
\subseteq L_{n}\!\left( D\right) $ and $\left| L_{n}\!\left( D_{n}\right)
\right| \leq \phi \left( \partial \right) $.

\item  For any $i\leq h\!\left( D\right) $, $\ell ^{\text{\textsc{f}}}\left(
L_{i}\!\left( D_{n}\right) \right) =\ell ^{\text{\textsc{f}}}\left(
L_{i}\!\left( D\right) \right) $. Thus $\partial _{n}^{\text{\textsc{c}}}$
and $\partial $\ have the same formulas, and hence $\phi \left( \partial
_{n}^{\text{\textsc{c}}}\right) =\phi \left( \partial \right) $.

\item  \textsc{e}$_{0}\!\left( D_{n}\right) \subseteq $ \textsc{e}$%
_{0}\!\left( D\right) $\ and $\Gamma _{\partial _{n}^{\text{\textsc{c}}%
}}=\Gamma _{\partial }$.
\end{enumerate}
\end{lemma}

\begin{proof}
By iteration, it will suffice to prove analogous assertions with respect to
every $\left( n,\alpha \right) $-collapsing involved. We skip trivial
conditions 2--4 and verify 1: $\partial _{n,\alpha }^{\text{\textsc{c}}%
}=\left\langle D_{n,\alpha },\text{\textsc{s}}_{n,\alpha },\ell _{n,\alpha
}^{\text{\textsc{f}}},\ell _{n,\alpha }^{\text{\textsc{g}}}\right\rangle \in 
\mathcal{D}_{n}^{\star }$. Consider the only nontrivial clause\ 3 of
Definition 6. It will suffice to show that \textsc{p}$\left( x,\!D_{n,\alpha
}\right) \subseteq \underset{y\in \text{\textsc{c}}\left( x,D_{n,\alpha
}\right) }{\bigcup }\ell _{n,\alpha }^{\text{\textsc{g}}}\!\left(
\left\langle x,y\right\rangle \!\right) $ holds for any $x\in \,$\textsc{v}$%
\left( D_{n,\alpha }\right) \subseteq \,$\textsc{v}$\left( D\right) $ such
that $\overrightarrow{\deg }\left( x,D_{n,\alpha }\right) >0$ and $%
\overleftarrow{\deg }\left( x,D_{n,\alpha }\right) >1$. If $h\left(
x,D\right) <n$ or $h\left( x,D\right) =n$ for $x\neq r$, then \textsc{p}$%
\left( x,\!D_{n,\alpha }\right) =\,$\textsc{p}$\left( x,\!D\right) $ and we
are done by the assumption \textsc{p}$\left( x,\!D\right) \subseteq 
\underset{y\in \text{\textsc{c}}\left( x,D\right) }{\bigcup }\ell ^{\text{%
\textsc{g}}}\!\left( \left\langle x,y\right\rangle \!\right) $ together with
clauses 4 (a), (b) of Definition 8. Otherwise we have $h\left( x,D\right) =n$
for $x=r$. Then every $z\in \,$\textsc{p}$\left( x,D_{n,\alpha }\right) $
determines a $u\in \,$\textsc{c}$\left( z,D\right) \cap S_{n,\alpha }$, and
hence $z\in \,$\textsc{p}$\left( u,D\right) $. Consider two cases.

\begin{enumerate}
\item  Suppose $\overleftarrow{\deg }\left( u,D\right) >1$. By the
assumption \textsc{p}$\left( u,\!D\right) \subseteq \underset{y\in \text{%
\textsc{c}}\left( u,D\right) }{\bigcup }\ell ^{\text{\textsc{g}}}\!\left(
\left\langle u,y\right\rangle \!\right) $ together with 4 (e)$_{i}$ of
Definition 8 this yields a $y\in \text{\textsc{c}}\left( u,D\right)
\subseteq \,$\textsc{c}$\left( x,D_{n,\alpha }\right) $ with $z\in \ell ^{%
\text{\textsc{g}}}\!\left( \left\langle u,y\right\rangle \!\right) \subseteq
\!\left( \left\langle x,y\right\rangle \!\right) $. Hence \textsc{p}$\left(
x,\!D_{n,\alpha }\right) \subseteq \underset{y\in \text{\textsc{c}}\left(
x,D_{n,\alpha }\right) }{\bigcup }\ell _{n,\alpha }^{\text{\textsc{g}}%
}\!\left( \left\langle x,y\right\rangle \!\right) $.

\item  Suppose $\overleftarrow{\deg }\left( u,D\right) =1$. Then $%
z=u_{\left( 1\right) }\in \ell ^{\text{\textsc{g}}}\!\left( \left\langle
u,y\right\rangle \!\right) \subseteq \ell _{n,\alpha }^{\text{\textsc{g}}%
}\!\left( \left\langle x,y\right\rangle \!\right) $ holds for any chosen $%
y\in \text{\textsc{c}}\left( u,D\right) \subseteq \,$\textsc{c}$\left(
x,D_{n,\alpha }\right) $ according to 4 (e)$_{ii}$ of Definition 8. Hence 
\textsc{p}$\left( x,\!D_{n,\alpha }\right) \subseteq \underset{y\in \text{%
\textsc{c}}\left( x,D_{n,\alpha }\right) }{\bigcup }\ell _{n,\alpha }^{\text{%
\textsc{g}}}\!\left( \left\langle x,y\right\rangle \!\right) $.
\end{enumerate}

This completes the proof of condition 1. Now consider 5 (with respect to
every $\left( n,\alpha \right) $-collapsing involved). \textsc{e}$%
_{0}\!\left( D_{n,\alpha }\right) \subseteq $ \textsc{e}$_{0}\!\left(
D\right) $\ is obvious, so it remains to establish $\Gamma _{\partial
_{n,\alpha }^{\text{\textsc{c}}}}=\Gamma _{\partial }$. In order to prove
the (more important) inclusion $\Gamma _{\partial _{n,\alpha }^{\text{%
\textsc{c}}}}\subseteq \Gamma _{\partial }$, it will suffice to show that
there is an assumption-preserving embedding of the open threads in $\partial
_{n,\alpha }^{\text{\textsc{c}}}$ into the open threads in $\partial $. So
let $\Theta _{n,\alpha }=\left[ v=x_{0},u=x_{1},\cdots ,x_{h\left( D\right)
}=\varrho \!\left( D\right) \right] \in \,$\textsc{th}$\left( e\!,\varrho
\!\left( D\right) \!,\partial _{n,\alpha }^{\text{\textsc{c}}}\right) $, $%
e=\left\langle u,v\right\rangle \in $ \textsc{e}$_{0}\!\left( D_{n,\alpha
}\right) $, be any given open thread in $\partial _{n,\alpha }^{\text{%
\textsc{c}}}$. A desired open thread in $\partial $, $\Theta =\left[
v^{\prime }=x_{0}^{\prime },u^{\prime }=x_{1}^{\prime },\cdots ,x_{h\left(
D\right) }^{\prime }=\varrho \!\left( D\right) \right] \in \,$\textsc{th}$%
\left( e^{\prime },\varrho \!\left( D\right) \!,\partial \right) $, $%
e^{\prime }=\left\langle u^{\prime },v^{\prime }\right\rangle \in $ \textsc{e%
}$_{0}\!\left( D\right) $ for $\ell _{n,\alpha }^{\text{\textsc{f}}}\left(
v^{\prime }\right) =\ell ^{\text{\textsc{f}}}\left( v\right) $ is defined by
cases as follows.

\begin{enumerate}
\item  Suppose $r\neq x_{i}$ for all $i\leq h\left( D\right) $. Then $\Theta
:=$ $\Theta _{n,\alpha }$, i.e. $\left( \forall i\leq h\left( D\right)
\right) x_{i}^{\prime }:=x_{i}$.

\item  Otherwise, $r=x_{m}$ and $\overleftarrow{\deg }\left(
x_{m},D_{n,\alpha }\right) >1$, where $m:=h\left( D\right) -n>0$. Consider
the following two subcases.

\begin{enumerate}
\item  Suppose $m>0$, i.e. $n<h\left( D\right) $, and note that $x_{m-1}\in
\,$\textsc{v}$\left( D\right) $ and $\overleftarrow{\deg }\left(
x_{m-1},D\right) =1$. Then let $x_{m}^{\prime }:=y$ such that \textsc{p}$%
\left( x_{m-1},\!D\right) =\left\{ y\right\} $. Note that $\ell ^{\text{%
\textsc{f}}}\left( x_{m}^{\prime }\right) =\ell _{n,\alpha }^{\text{\textsc{f%
}}}\left( x_{m}\right) $. For all $i\neq m$ let $x_{i}^{\prime }:=x_{i}$.

\item  Let $m=0$, i.e. $n=h\left( D\right) $. If $\overleftarrow{\deg }%
\left( x_{i},D\right) =1$ for all $0<i<h\left( D\right) $, then let $\Theta
:=$ $\Theta _{n,\alpha }$. Otherwise, let $j:=\min \left\{ i>0:%
\overleftarrow{\deg }\left( x_{i},D\right) >1\right\} $. Then let $%
x_{0}^{\prime }$ be any $v^{\prime }\in \,$\textsc{c}$\left( u,D\right) \cap
S_{n,\alpha }$\ such that $x_{j+1}\in \ell ^{\text{\textsc{g}}}\left(
\left\langle u,v^{\prime }\right\rangle \right) $. Clearly $\ell ^{\text{%
\textsc{f}}}\left( x_{0}^{\prime }\right) =\ell _{n,\alpha }^{\text{\textsc{f%
}}}\left( x_{0}\right) $. For all $i>0$ let $x_{i}^{\prime }:=x_{i}$.
\end{enumerate}
\end{enumerate}

This completes our definition of $\Theta $. That $\Theta $ is an open thread
is easily verified using definition of $\ell _{n,\alpha }^{\text{\textsc{g}}%
} $ (see Definition 8 (4)). Thus $\Gamma _{\partial _{n,\alpha }^{\text{%
\textsc{c}}}}\subseteq \Gamma _{\partial }$. $\Gamma _{\partial }\subseteq
\Gamma _{\partial _{n,\alpha }^{\text{\textsc{c}}}}$ is proved analogously
by inversion $\Theta \hookrightarrow \Theta _{n,\alpha }$ that is defined by
substituting $r$ for (at most one) $x_{m}\in S_{n,\alpha }\setminus \left\{
r\right\} $. This completes the whole proof.
\end{proof}

\subsubsection{Horizontal compressing}

As mentioned above, horizontal compression $\partial \hookrightarrow
\partial ^{\text{\textsc{c}}}$ is obtained by bottom-up iteration of
horizontal collapsing $\partial \hookrightarrow \partial _{n}^{\text{\textsc{%
c}}}$, $n\leq h\left( \partial \right) $. For the sake of brevity we
consider tree-like inputs $\partial \in \mathcal{T}^{\ast }$.

\begin{definition}[horizontal compressing]
For any given $\partial \in \mathcal{T}^{\ast }$ denote by $\partial ^{\text{%
\textsc{c}}}\in \mathcal{D}^{\star }$ the last deduction in the following
iteration chain 
\begin{equation*}
\partial =\partial _{\left( 0\right) }^{\text{\textsc{c}}},\ \partial
_{\left( 1\right) }^{\text{\textsc{c}}},\ \cdots ,\ \partial _{\left(
h\left( \partial \right) \right) }^{\text{\textsc{c}}}=\partial ^{\text{%
\textsc{c}}}
\end{equation*}
where for every $i<h\left( \partial \right) $ we let $\partial _{\left(
i+1\right) }^{\text{\textsc{c}}}:=\left( \partial _{\left( i\right) }^{\text{%
\textsc{c}}}\right) _{i+1}^{\text{\textsc{c}}}$. It is readily seen that all 
$\partial ^{\text{\textsc{c}}}$ in question are mutually isomorphic
(actually equal up to the choice of $r\in S_{n,\alpha }$). The operation $%
\partial \hookrightarrow \partial ^{\text{\textsc{c}}}$ is called the \emph{%
horizontal dag-like compression}, in \textsc{NM}$_{\rightarrow }^{\star }$.
\end{definition}

\begin{theorem}
For any tree-like deduction $\partial \in \mathcal{T}^{\ast }$ with
root-formula $\rho $, the horizontal compression $\partial ^{\text{\textsc{c}%
}}$ is a plain dag-like \textsc{NM}$_{\rightarrow }$ deduction of $\rho $
from the same assumptions $\Gamma _{\partial ^{\text{\textsc{c}}}}=\Gamma
_{\partial }$. Moreover $\left| \partial ^{\text{\textsc{c}}}\right| \leq
h\left( \partial \right) \times \phi \left( \partial \right) $ and $\mu
\left( \partial ^{\text{\textsc{c}}}\right) =\mu \left( \partial \right) $.\
In particular, if $\Gamma _{\partial }=\emptyset $ and $h\left( \partial
\right) $, $\phi \left( \partial \right) $, $\mu \left( \partial \right) $
are polynomial in $\left| \rho \right| $, then $\partial ^{\text{\textsc{c}}%
} $ is a plain dag-like \textsc{NM}$_{\rightarrow }$ proof of $\rho $ whose
size and weight are polynomial in $\left| \rho \right| $.
\end{theorem}

\begin{proof}
Let $\partial =\left\langle T,\text{\textsc{s}},\ell ^{\text{\textsc{f}}%
},\ell ^{\text{\textsc{g}}}\text{\textsc{\thinspace }}\right\rangle \in 
\mathcal{T}^{\ast }$ and $\partial _{n}^{\text{\textsc{c}}}=\left\langle
D_{n},\text{\textsc{s}}_{n},\ell _{n}^{\text{\textsc{f}}},\ell _{n}^{\text{%
\textsc{g}}}\text{\textsc{\thinspace }}\right\rangle $ for $n\leq h\!\left(
D\right) $. By Lemma 9 (2, 3) we have 
\begin{eqnarray*}
\left| \partial ^{\text{\textsc{c}}}\right| &=&\overset{h\left( T\right) }{%
\underset{n=0}{\bigcup }}\left| L_{n}\left( D_{n}\right) \right| \leq \\
1+2+\overset{h\left( T\right) }{\underset{n=2}{\bigcup }}\left| L_{n}\left(
D_{n}\right) \right| &\leq &3+\left( h\left( T\right) -1\right) \cdot \phi
\left( \partial \right) < \\
h\left( T\right) \cdot \phi \left( \partial \right) &=&h\left( \partial
\right) \times \phi \left( \partial \right)
\end{eqnarray*}
as required. The rest immediately follows from Lemma 9 (1, 4, 5) by
induction on $n\leq h\left( T\right) $.
\end{proof}

Together with Theorem 4 and Lemma 5 this yields

\begin{corollary}
Any given minimal tautology $\rho $ has a plain dag-like \textsc{NM}$%
_{\rightarrow }$ proof $\partial ^{\text{\textsc{c}}}$ whose size and weight
are polynomial in $\left| \rho \right| $. Actually the following holds. 
\begin{equation*}
\fbox{$\left| \partial ^{\text{\textsc{c}}}\right| <18\left| \rho \right|
\left( \left| \rho \right| +1\right) ^{2}\left( \left| \rho \right|
+2\right) $ $=$ $\mathcal{O}\left( \left| \rho \right| ^{4}\right) $ and $%
\left\| \partial ^{\text{\textsc{c}}}\right\| =\mathcal{O}\left( \left| \rho
\right| ^{5}\right) $}
\end{equation*}
\end{corollary}

\begin{example}
\footnote{{\footnotesize See Appendices B, C for more sophisticated examples.%
}} Consider a following (tree-like) \textsc{NM}$_{\rightarrow }$ deduction $%
\partial $.

$\partial =\fbox{$\dfrac{\dfrac{\dfrac{\alpha \ \ \quad \alpha \rightarrow
\rho }{\rho \quad \quad }\ \ \quad \QDATOP{{}}{\rho \rightarrow \alpha }}{%
\alpha }\ \quad \dfrac{\dfrac{\left[ \alpha \right] \ \ \quad \alpha
\rightarrow \rho }{\rho \quad \quad }}{\alpha \rightarrow \rho \quad }}{\rho
\quad \quad }$}$

$\ \ \cong \fbox{$\dfrac{\dfrac{\dfrac{\alpha \ \ \quad \alpha \rightarrow
\rho }{\rho \quad \quad }\ \ \quad \dfrac{\rho \rightarrow \alpha }{\rho
\rightarrow \alpha }}{\alpha }\ \quad \dfrac{\dfrac{\left[ \alpha \right] \
\ \quad \alpha \rightarrow \rho }{\rho \quad \quad }}{\alpha \rightarrow
\rho \quad }}{\rho \quad \quad }$}\in \mathcal{T}^{\ast }$

(As usual $\left[ \alpha \right] $ indicates that the right-hand side
assumption $\alpha $ is discharged by $\left( \rightarrow I\right) :\dfrac{%
\rho }{\alpha \rightarrow \rho }$ occurring below.) Horizontally compressed
(re)dag-like \textsc{NM}$_{\rightarrow }^{\star }$ deduction $\partial ^{%
\text{\textsc{c}}}$ arises by successively merging two nodes with label $%
\rho $ and two identical pairs of assumptions $\alpha $,\ $\alpha
\rightarrow \rho $.\smallskip

$\partial ^{\text{\textsc{c}}}=\fbox{$\dfrac{\dfrac{\dfrac{\rho \rightarrow
\alpha }{\rho \rightarrow \alpha }\quad \dfrac{\alpha \qquad \alpha
\rightarrow \rho }{\quad \rho \quad \quad \quad \quad }\quad }{\quad \quad
\alpha \quad \qquad \qquad \quad \quad \alpha \rightarrow \rho }}{\rho \quad
\quad }$}$

$\quad \ \cong \fbox{$\dfrac{\dfrac{\QDATOP{{}}{\rho \rightarrow \alpha }%
\quad \dfrac{\alpha \qquad \alpha \rightarrow \rho }{\quad \rho \quad \quad
\quad \quad }\quad }{\quad \quad \alpha \quad \qquad \qquad \quad \quad
\alpha \rightarrow \rho }}{\rho \quad \quad }$}\smallskip \in \,$\textsc{NM}$%
_{\rightarrow }^{\star }$

Clearly $\partial ^{\text{\textsc{c}}}$ is not a tree, as it contains a
``diamond'' $\alpha \QDATOP{\overset{\rho }{\nearrow \nwarrow }}{\underset{%
\rho }{\nwarrow \nearrow }}\alpha \rightarrow \rho $ .

Note that $\Gamma _{\partial }=\left\{ \alpha ,\alpha \rightarrow \rho ,\rho
\rightarrow \alpha \right\} $, as the left-hand side assumption $\alpha $ is
open in a tree-like deduction $\partial $. No consider the compressed redag $%
\partial ^{\text{\textsc{c}}}$. We have\smallskip

$\partial ^{\text{\textsc{c}}}=\fbox{$\dfrac{\dfrac{\dfrac{u_{1}:\rho
\rightarrow \alpha }{v_{1}:\rho \rightarrow \alpha }\quad \dfrac{%
u_{2}:\alpha \qquad u_{3}:\alpha \rightarrow \rho }{\quad v_{2}:\rho \quad
\quad \quad \quad }\quad }{\quad \quad w_{1}:\alpha \quad \qquad \qquad
\quad \quad w_{2}:\alpha \rightarrow \rho }}{z:\rho \quad \quad }$}%
\smallskip $

($u_{i},v_{j},w_{j}$ and $z=\varrho \left( \partial ^{\text{\textsc{c}}%
}\right) $ being the underlying nodes). Except for $v_{2}$ and $z$, all
nodes have exactly one parent, while $\ell ^{\text{\textsc{g}}}\left(
\left\langle v_{2},u_{2}\right\rangle \right) =$ $\ell ^{\text{\textsc{g}}%
}\left( \left\langle v_{2},u_{3}\right\rangle \right) =\left\{
w_{1},w_{2}\right\} $, i.e. both leaves $u_{2}$ and $u_{3}$ have two $\ell ^{%
\text{\textsc{g}}}$-grandparents $w_{1}$ and $w_{2}$ (which are inherited
from standard tree-like grandparents of the first and the last top nodes in $%
\partial $). This yields 5 deduction threads in $\partial ^{\text{\textsc{c}}%
}$: $\left\{ u_{1},v_{1},w_{1},z\right\} $, $\left\{
u_{2},v_{2},w_{1},z\right\} $, $\left\{ u_{2},v_{2},w_{2},z\right\} $, $%
\left\{ u_{3},v_{2},w_{1},z\right\} $, $\left\{ u_{3},v_{2},w_{2},z\right\} $
and $\ell ^{\text{\textsc{f}}}$-threads $\left\{ \rho \rightarrow \alpha
,\rho \rightarrow \alpha ,\alpha ,\rho \right\} $, $\left\{ \alpha ,\rho
,\alpha ,\rho \right\} $, $\left\{ \alpha ,\rho ,\alpha \rightarrow \rho
,\rho \right\} $, $\left\{ \alpha \rightarrow \rho ,\rho ,\alpha ,\rho
\right\} $, $\left\{ \alpha \rightarrow \rho ,\rho ,\alpha \rightarrow \rho
,\rho \right\} $, while $\alpha $ is open in $\Theta =\left\{
u_{2},v_{2},w_{1},z\right\} $ due to $\ell ^{\text{\textsc{f}}}\left( \Theta
\right) =\left\{ \alpha ,\rho ,\alpha ,\rho \right\} $ (that other
assumptions $\alpha \rightarrow \rho $, $\rho \rightarrow \alpha $ are open
in $\partial ^{\text{\textsc{c}}}$ is readily seen). Hence $\Gamma
_{\partial ^{\text{\textsc{c}}}}=\left\{ \alpha ,\alpha \rightarrow \rho
,\rho \rightarrow \alpha \right\} =\Gamma _{\partial }$, i.e. $\partial $
and $\partial ^{\text{\textsc{c}}}$ are deductions of $\rho $ from the
assumptions $\left\{ \alpha ,\alpha \rightarrow \rho ,\rho \rightarrow
\alpha \right\} $, although at the first glance $\alpha $ seems to be
discharged in $\partial ^{\text{\textsc{c}}}$.
\end{example}

\subsection{Dag-to-tree unfolding in \textsc{NM}$_{\rightarrow }$}

We learned that all minimal propositional tautologies are provable by plain
dag-like \textsc{NM}$_{\rightarrow }$ deductions of ``small'' size, but at
the moment we don't know whether underlying \textsc{NM}$_{\rightarrow
}^{\star }$ provability infers validity in minimal logic. The affirmative
answer follows by dag-to-tree unfolding, to be thought of as inversion of
the tree-to-dag compression under consideration. The unfolded tree-like
deduction $\partial ^{\text{\textsc{u}}}$ is defined by descending recursion
on the height of a given \textsc{NM}$_{\rightarrow }^{\star }$ deduction $%
\partial $ such that for any $n\leq h\left( \partial \right) $, the $n^{th}$%
\ horizontal section of $\partial ^{\text{\textsc{u}}}$ is obtained by
splitting previously obtained nodes $v$, $h\left( v\right) =n$, having $p$
parents, $u_{1},\cdots u_{p}$, $p>1$, into $p$ new copies\ $v_{1},\cdots
v_{p}$. Previously obtained (tree-like!) successors of\ $v$\ are separated
according to the underlying assignment $\ell ^{\text{\textsc{g}}}$ such that
for every $0<i<p$, $u_{i}$ becomes the only parent of $v_{i}$. Moreover, if $%
T$ is the old tree rooted in $v$, then every $v_{i}$ $\left( 0<i<p\right) $\
becomes the root of a maximal subtree of $T$ whose leaves are $\ell ^{\text{%
\textsc{g}}}$-grandchildren of $u_{i}$ (i.e. $u_{i}$ is a $\ell ^{\text{%
\textsc{g}}}$-grandparent of every leaf in question). Except for the $\ell ^{%
\text{\textsc{g}}}$-related separation this is just standard graph theoretic
dag-to-tree unfolding (see below a precise definition).

\begin{definition}
Consider any $\partial =\left\langle D,\text{\textsc{s}},\ell ^{\text{%
\textsc{f}}},\ell ^{\text{\textsc{g}}}\text{\textsc{\thinspace }}%
\right\rangle \in \mathcal{D}_{n}^{\star }$, $n\leq h\!\left( D\right) $ and
a fixed $r\in L_{n}\left( D\right) $ with $p:=\left| \text{\textsc{p}}\left(
r,D\right) \right| >1$.\ We define $\left( n,r\right) $-unfolded deduction $%
\partial _{n,r}^{\text{\textsc{u}}}=\left\langle D_{n,r},\text{\textsc{s}}%
_{n,r},\ell _{n,r}^{\text{\textsc{f}}},\ell _{n,r}^{\text{\textsc{g}}%
}\right\rangle \in \mathcal{D}_{n}^{\star }$ that arises by tree-like
unfolding of $r$, as follows. Let $r_{1},\cdots ,r_{p}\notin \,$\textsc{v}$%
\left( D\right) $ be a fixed collection of new vertices and $\left( D\right)
_{r_{1}},\cdots ,\left( D\right) _{r_{p}}$ the corresponding collection of
disjoint copies of $\left( D\right) _{r}$. Let $\varepsilon :\left[ p\right]
\rightarrow \,$\textsc{p}$\left( r,D\right) $ be a fixed 1--1 enumeration of 
\textsc{p}$\left( r,D\right) $. Then for any $i\in \left[ p\right] $ we
denote by $\left( D\right) _{i}^{-}$ a subtree of $\left( D\right) _{r_{i}}$
that is obtained by deleting the (copies of) subtrees $\left( D\right) _{y}$%
, for all $\left\langle x,y\right\rangle \in \,$\textsc{e}$\left( \left(
D\right) _{r_{i}}\right) $ such that $\varepsilon \left( i\right) \notin
\ell ^{\text{\textsc{g}}}\left( \left\langle x,y\right\rangle \right) $.
Furthermore, we denote by $\left[ \left( D\right) _{i}^{-}\right] $ a tree
that extends $\left( D\right) _{i}^{-}$ by a new root $\varepsilon \left(
i\right) $; thus $\varrho \left( \left( D\right) _{i}^{-}\right) =r_{i}$ and 
$\varrho \left( \left[ \left( D\right) _{i}^{-}\right] \right) =\varepsilon
\left( i\right) $ with $\left\{ \varepsilon \left( i\right) \right\} =\,$%
\textsc{p}$\left( r_{i},\left[ \left( D\right) _{i}^{-}\right] \right) $.
Having this we stipulate:

\begin{enumerate}
\item  $D_{n,r}$\ arises from $D$ by deleting $\left( D\right) _{r}$ and
replacing every remaining node $\varepsilon \left( i\right) \in \,$\textsc{p}%
$\left( r,D\right) $ by the whole subtree $\left[ \left( D\right) _{i}^{-}%
\right] $.

That is, \textsc{v}$\left( D_{n,r}\right) :=\ \left( \text{\textsc{v}}\left(
D\right) \setminus \text{\textsc{v}}\left( \left( D\right) _{r}\right)
\right) \cup \underset{i=1}{\overset{p}{\bigcup }}\text{\textsc{v}}\left(
\left( D\right) _{i}^{-}\right) $. The edges are given by \textsc{%
e\thinspace }$\left( D_{n,r}\right) \!:=\left( \text{\textsc{e\negthinspace
\thinspace }}\left( D\right) \setminus \text{\textsc{e}}\left( \left(
D\right) _{r}\right) \right) \cup \underset{i=1}{\overset{p}{\bigcup }}%
\left( \text{\textsc{e}}\left( \left( D\right) _{i}^{-}\right) \cup
\left\langle \varepsilon \left( i\right) ,r_{i}\right\rangle \right) $.$\!$

\item  For any $u\in \,$\textsc{v}$\left( D_{n,r}\right) $ we define \textsc{%
s}$_{n,r}\!\left( u,D_{n,r}\right) $ by cases as follows.

\begin{enumerate}
\item  If $u\notin \underset{i=1}{\overset{p}{\bigcup }}$\textsc{v}$\left(
\left( D\right) _{i}^{-}\right) \cup \,\!\,$\textsc{p}$\left( r,D\right) $,
then \textsc{s}$_{n,r}\!\left( u,D_{n,r}\right) :=\,$\textsc{s}$\left(
u,D\right) $.

\item  If $u\in \underset{i=1}{\overset{p}{\bigcup }}$\textsc{v}$\left(
\left( D\right) _{i}^{-}\right) $, then \textsc{s}$_{n,r}\left(
u,D_{n,r}\right) :=\,$\textsc{s}$\left( u,D\right) $ (modulo isomorphism).

\item  For any $i\in \left[ 1,p\right] $ we let \textsc{s}$_{n,r}\!\left(
\varepsilon \left( i\right) ,D_{n,r}\right) :=X_{i}\cup Y_{i}$, where

$X_{i}=\left\{ y\in L_{n}\!\left( D_{n,r}\right) :\!\!\left. 
\begin{array}{c}
\left( \exists x\in \text{\textsc{s}}\left( \varepsilon \left( i\right)
,D\right) \right) \\ 
\left( x=y\vee \left( x=r\wedge y=r_{i}\right) \right)
\end{array}
\!\right. \!\!\right\} \ $and

$Y_{i}=\left\{ \left\langle y_{0},y_{1}\right\rangle \in L_{n}\!\left(
D_{n,r}\right) ^{2}:\!\!\left. 
\begin{array}{c}
\left( \exists \left\langle x_{0},x_{1}\right\rangle \in \text{\textsc{s}}%
\left( \varepsilon \left( i\right) ,D\right) \right) \left( \forall j\leq
1\right) \\ 
\left( x_{j}=y_{j}\vee \left( x_{j}=r\wedge y_{j}=r_{i}\right) \right)
\end{array}
\!\right. \!\!\right\} $.
\end{enumerate}

\item  For any $u\in \,$\textsc{v}$\left( D_{n,r}\right) $ we let $\ell
_{n,r}^{\text{\textsc{f}}}\left( u\right) :=\ell ^{\text{\textsc{f}}}\left( 
\widehat{u}\right) $, where $\widehat{u}\in \,$\textsc{v}$\left( D\right) $
is a (uniquely determined) preimage of $u$ in $D$.

\item  For any $e=\left\langle u,v\right\rangle \in \,$\textsc{e}$%
_{\triangle }\!\left( D_{n,r}\right) $ we define $\ell _{n,r}^{\text{\textsc{%
g}}}\left( e\right) $ by cases as follows, while without loss of generality
assuming that $v\in \,$\textsc{L}$\left( D_{n,r}\right) $ or $\overleftarrow{%
\deg }\left( v,D_{n,r}\right) >1$ (cf. analogous passage in Definition 8).

\begin{enumerate}
\item  If $h\!\left( u,D_{n,r}\right) \in \left[ 0,n-2\right] $, or $%
h\!\left( u,D_{n,r}\right) =n-1$ and $v\notin \left\{ r_{1},\cdots
,r_{p}\right\} $, or else $e\in \,$\textsc{e}$_{0}\left( D_{n,r}\right) $
with $h\!\left( u,D_{n,r}\right) \geq n$ and $\left( \forall i\in \left[ p%
\right] \right) r_{i}\npreceq _{D_{n,r}}\!u$, then $\ell _{n,r}^{\text{%
\textsc{g}}}\left( e\right) :=\ell ^{\text{\textsc{g}}}\left( e\right) $.

\item  Otherwise, if $v=r_{i}$ (hence $u=\varepsilon \left( i\right) $), or
else $e\in \,$\textsc{e}$_{0}\left( D_{n,r}\right) $ with $h\!\left(
u,D_{n,r}\right) \geq n$ and $r_{i}\preceq _{D_{n,r}}\!u$, then $\ell
_{n,r}^{\text{\textsc{g}}}\!\left( e\right) :=\ell ^{\text{\textsc{g}}%
}\left( \left\langle \varepsilon \left( i\right) ,r\right\rangle \right) $.
\end{enumerate}
\end{enumerate}

To complete the $\left( n,r\right) $\emph{-unfolding operation} $\partial
\hookrightarrow \partial _{n,r}^{\text{\textsc{u}}}$, we let $\partial
_{n,r}^{\text{\textsc{u}}}:=\partial $ in the case $\left| \text{\textsc{p}}%
\left( r,D\right) \right| =1$. Now let $\partial _{n}^{\text{\textsc{u}}}$
arise from $\partial $ by applying $\left( n,r\right) $-unfolding
successively to all $r\in L_{n}\left( D\right) $. That is, $\partial _{n}^{%
\text{\textsc{u}}}$ is the iteration of $\partial _{n,r}^{\text{\textsc{u}}}$%
\ with respect to all nodes $r$ occurring in the $n^{th}$\ horizontal
section of $D$. The operation $\partial \hookrightarrow \partial _{n}^{\text{%
\textsc{u}}}$ is called the \emph{horizontal unfolding on level} $n$, in 
\textsc{NM}$_{\rightarrow }^{\star }$.
\end{definition}

\begin{lemma}
For any $\partial =\left\langle D,\text{\textsc{s}},\ell ^{\text{\textsc{f}}%
},\ell ^{\text{\textsc{g}}}\text{\textsc{\thinspace }}\right\rangle \in 
\mathcal{D}_{n}^{\star }$ and $\partial _{n}^{\text{\textsc{u}}%
}=\left\langle D_{n},\text{\textsc{s}}_{n},\ell _{n}^{\text{\textsc{f}}%
},\ell _{n}^{\text{\textsc{g}}}\text{\textsc{\thinspace }}\right\rangle $, $%
n\leq h\!\left( D\right) $, the following conditions 1--5 hold.

\begin{enumerate}
\item  $\partial _{n}^{\text{\textsc{u}}}\in \mathcal{D}_{n-1}^{\star }$.

\item  $\varrho \!\left( D_{n}\right) =\varrho \!\left( D\right) $ and $%
h\!\left( D_{n}\right) =h\!\left( D\right) $.

\item  For any $i<n$, $L_{i}\!\left( D_{n}\right) =L_{i}\!\left( D\right) $,
while $L_{n}\!\left( D_{n}\right) \supseteq L_{n}\!\left( D\right) $.

\item  For any $i<n<j$, $\ell ^{\text{\textsc{f}}}\left( L_{i}\!\left(
D_{n}\right) \right) =\ell ^{\text{\textsc{f}}}\left( L_{i}\!\left( D\right)
\right) $ and $\ell ^{\text{\textsc{f}}}\left( L_{j}\!\left( D_{n}\right)
\right) \subseteq \ell ^{\text{\textsc{f}}}\left( L_{j}\!\left( D\right)
\right) $, while $\ell ^{\text{\textsc{f}}}\left( L_{n}\!\left( D_{n}\right)
\right) =\ell ^{\text{\textsc{f}}}\left( L_{n}\!\left( D\right) \right) $.
Hence $\phi \left( \partial _{n}^{\text{\textsc{u}}}\right) \subseteq \phi
\left( \partial \right) $.

\item  $\Gamma _{\partial _{n}^{\text{\textsc{u}}}}\subseteq \Gamma
_{\partial }$.
\end{enumerate}
\end{lemma}

\begin{proof}
By iteration, it will suffice to prove analogous assertions with respect to
every $\left( n,r\right) $-unfolding involved. We skip trivial conditions
2--4 and verify 1: $\partial _{n,r}^{\text{\textsc{u}}}=\left\langle D_{n,r},%
\text{\textsc{s}}_{n,r},\ell _{n,r}^{\text{\textsc{f}}},\ell _{n,r}^{\text{%
\textsc{g}}}\text{\textsc{\thinspace }}\right\rangle \in \mathcal{D}%
_{n-1}^{\star }$. First of all we observe that every subtree $\left[ \left(
D\right) _{i}^{-}\right] $ that replaced $\varepsilon \left( i\right) \in \,$%
\textsc{p}$\left( r,D\right) $ according to clause 1 of Definition 14
represents a (tree-like) \textsc{NM}$_{\rightarrow }^{\ast }$ deduction of $%
\ell ^{\text{\textsc{f}}}\left( \varepsilon \left( i\right) \right) $ such
that $h\left( \left[ \left( D\right) _{i}^{-}\right] \right) =1+h\left(
\left( D\right) _{i}^{-}\right) =1+h\left( \left( D\right) _{r}\right) $.
This easily follows by induction on $h\left( \left( D\right) _{r}\right) $
using clause 2 of Definition 6 with respect to $\partial $. So $\partial
_{n,r}^{\text{\textsc{u}}}$ is structurally well-defined. To complete the
proof of local correctness consider the only nontrivial clause\ 3 of
Definition 6 with respect to $\partial _{n,r}^{\text{\textsc{u}}}$. It will
suffice to show that \textsc{p}$\left( x,\!D_{n,r}\right) \subseteq 
\underset{y\in \text{\textsc{c}}\left( x,D_{n,r}\right) }{\bigcup }\ell
_{n,r}^{\text{\textsc{g}}}\!\left( \left\langle x,y\right\rangle \!\right) $
holds for any $x\in \,$\textsc{v}$\left( D_{n,r}\right) $ such that $%
\overrightarrow{\deg }\left( x,D_{n,r}\right) >0$ and $\overleftarrow{\deg }%
\left( x,D_{n,r}\right) >1$. If $h\left( x,D\right) <n$ for $r\notin \,$%
\textsc{c}$\left( x,D\right) $, or $h\left( x,D\right) =n$ for $x\neq r_{i}$
($1\leq i\leq p$), then \textsc{p}$\left( x,\!D_{n,r}\right) =\,$\textsc{p}$%
\left( x,\!D\right) $ and we are done by the assumption \textsc{p}$\left(
x,\!D\right) \subseteq \underset{y\in \text{\textsc{c}}\left( x,D\right) }{%
\bigcup }\ell ^{\text{\textsc{g}}}\!\left( \left\langle x,y\right\rangle
\!\right) $. Consider the remaining cases.

\begin{enumerate}
\item  Suppose $x=r_{i}$ ($1\leq i\leq p$). We have \textsc{p}$\left(
x,\!D_{n,r}\right) =\left\{ \varepsilon \left( i\right) \right\} \subseteq
\, $\textsc{p}$\left( r,\!D\right) $. Moreover, by the assumption \textsc{p}$%
\left( r,\!D\right) \subseteq \underset{y\in \text{\textsc{c}}\left(
r,D\right) }{\bigcup }\ell ^{\text{\textsc{g}}}\!\left( \left\langle
r,y\right\rangle \!\right) $, there exists a $y\in \,$\textsc{c}$\left(
r,D\right) $ with $\varepsilon \left( i\right) \in \ell ^{\text{\textsc{g}}%
}\left( \left\langle r,y\right\rangle \right) $. From this, by the
definition of $\left( D\right) _{i}^{\prime }$ and clauses\ 2 (b), (c) of
Definition 6 with respect to $\partial _{n,r}^{\text{\textsc{u}}}$, we
arrive at $y\in \,$\textsc{c}$\left( x,\left( D\right) _{i}^{\prime }\right)
\subseteq \,$\textsc{c}$\left( x,\!D_{n,r}\right) $ and $\varepsilon \left(
i\right) \in \ell _{n,r}^{\text{\textsc{g}}}\left( \left\langle
x,y\right\rangle \right) $. Thus \textsc{p}$\left( x,\!D_{n,r}\right)
\subseteq \underset{y\in \text{\textsc{c}}\left( x,D_{n,r}\right) }{\bigcup }%
\ell _{n,r}^{\text{\textsc{g}}}\!\left( \left\langle x,y\right\rangle
\!\right) $.

\item  Suppose $r\in \,$\textsc{c}$\left( x,D\right) $, and hence $%
x=u_{i}=\varepsilon \left( i\right) $ for some $1\leq i\leq p$. Consider any 
$z\in \,$\textsc{p}$\left( x,\!D_{n,r}\right) =\,$\textsc{p}$\left(
x,\!D\right) \subseteq \underset{y\in \text{\textsc{c}}\left( x,D\right) }{%
\bigcup }\ell ^{\text{\textsc{g}}}\!\left( \left\langle x,y\right\rangle
\!\right) $ and let $z\in \ell ^{\text{\textsc{g}}}\!\left( \left\langle
x,y\right\rangle \!\right) $ for some $y\in \,$\textsc{c}$\left( x,D\right) $%
. If $y\neq r$\ then $y\in \,$\textsc{c}$\left( x,D_{n,r}\right) $ and we
are done. Otherwise $y=r$, and then by clause 4 (b) of Definition 14 we
arrive at $x=\varepsilon \left( i\right) \in \ell ^{\text{\textsc{g}}}\left(
\left\langle \varepsilon \left( i\right) ,r\right\rangle \right) =\ell
_{n,r}^{\text{\textsc{g}}}\!\left( \left\langle \varepsilon \left( i\right)
,r_{i}\right\rangle \right) $ with $r_{i}\in \,$\textsc{c}$\left(
x,D_{n,r}\right) $. Hence \textsc{p}$\left( x,\!D_{n,r}\right) \subseteq 
\underset{y\in \text{\textsc{c}}\left( x,D_{n,r}\right) }{\bigcup }\ell
_{n,r}^{\text{\textsc{g}}}\!\left( \left\langle x,y\right\rangle \!\right) $.
\end{enumerate}

This completes the proof of condition 1. Now consider 5 (with respect to
every $\left( n,r\right) $-unfolding involved). In order to prove $\Gamma
_{\partial _{n,r}^{\text{\textsc{u}}}}\subseteq \Gamma _{\partial }$, it
will suffice to show that there is an assumption-preserving embedding of the
open threads in $\partial _{n,r}^{\text{\textsc{u}}}$ into the open threads
in $\partial $. So let $\Theta _{n,r}=\left[ v=x_{0},u=x_{1},\cdots
,x_{h\left( D\right) }=\varrho \!\left( D\right) \right] $ $\in \,$\textsc{th%
}$\left( e\!,\varrho \!\left( D\right) \!,\partial _{n,r}^{\text{\textsc{u}}%
}\right) $, $e=\left\langle u,v\right\rangle \in $ \textsc{e}$_{0}\!\left(
D_{n,r}\right) $, be any given open thread in $\partial _{n,r}^{\text{%
\textsc{u}}}$. (We consider only the proper case $\overleftarrow{\deg }%
\left( r,D\right) >1$.) A desired open thread in $\partial $, $\Theta =\left[
v^{\prime }=x_{0}^{\prime },u^{\prime }=x_{1}^{\prime },\cdots ,x_{h\left(
D\right) }^{\prime }=\varrho \!\left( D\right) \right] \in \,$\textsc{th}$%
\left( e^{\prime },\varrho \!\left( D\right) \!,\partial \right) $, $%
e^{\prime }=\left\langle u^{\prime },v^{\prime }\right\rangle \in $ \textsc{e%
}$_{0}\!\left( D\right) $ for $\ell _{n,\alpha }^{\text{\textsc{f}}}\left(
v^{\prime }\right) =\ell ^{\text{\textsc{f}}}\left( v\right) $ is obtained
by substituting $r$ for any $r_{i}$ occurring in $\Theta $. That is, for any 
$j\leq h\left( D\right) $ we let $x_{j}^{\prime }:=r$, if $j=n$ and $%
x_{j}\neq r$, else $x_{j}^{\prime }:=x_{j}$. That $\Theta \in \,$\textsc{th}$%
\left( e^{\prime },\varrho \!\left( D\right) \!,\partial \right) $ and $%
e^{\prime }=\left\langle x_{1}^{\prime },x_{0}^{\prime }\right\rangle \in $ 
\textsc{e}$_{0}\!\left( D\right) $ easily follows by the definition of $\ell
_{n,r}^{\text{\textsc{g}}}$. Hence $\Gamma _{\partial _{n,r}^{\text{\textsc{u%
}}}}\subseteq \Gamma _{\partial }$ .This completes the whole proof by
iteration with respect to all $r\in L_{n}\left( D\right) $, $\overleftarrow{%
\deg }\left( r,D\right) >1$\ involved.
\end{proof}

\begin{definition}[horizontal unfolding]
For any given $\partial \in \mathcal{D}^{\star }$ denote by $\partial ^{%
\text{\textsc{u}}}\in \mathcal{T}^{\ast }$ the last deduction in the
following iteration chain 
\begin{equation*}
\partial =\partial _{\left( h\left( \partial \right) \right) }^{\text{%
\textsc{u}}},\ \partial _{\left( h\left( \partial \right) -1\right) }^{\text{%
\textsc{u}}},\ \cdots ,\ \partial _{\left( 0\right) }^{\text{\textsc{u}}%
}=\partial ^{\text{\textsc{u}}}
\end{equation*}
where for every $i<h\left( \partial \right) $ we let $\partial _{\left(
i-1\right) }^{\text{\textsc{u}}}:=\left( \partial _{\left( i\right) }^{\text{%
\textsc{u}}}\right) _{i-1}^{\text{\textsc{u}}}$. It is readily seen that all 
$\partial ^{\text{\textsc{u}}}$ in question are mutually isomorphic
(actually equal up to the enumerations $\varepsilon $). The operation $%
\partial \hookrightarrow \partial ^{\text{\textsc{u}}}$ is called the \emph{%
horizontal unfolding}, in \textsc{NM}$_{\rightarrow }^{\star }$.
\end{definition}

\begin{theorem}
For any dag-like \textsc{NM}$_{\rightarrow }^{\star }$ deduction $\partial $
with root-formula $\rho $, the horizontal unfolding $\partial ^{\text{%
\textsc{u}}}$ is a tree-like \textsc{NM}$_{\rightarrow }^{\ast }$ deduction
of $\rho $ such that $\Gamma _{\partial ^{\text{\textsc{u}}}}\subseteq
\Gamma _{\partial }$. In particular, if $\partial $ is a plain dag-like 
\textsc{NM}$_{\rightarrow }$ proof of $\rho $, then $\partial ^{\text{%
\textsc{u}}}$ is a tree-like \textsc{NM}$_{\rightarrow }^{\ast }$\ proof of $%
\rho $.
\end{theorem}

\begin{proof}
The assertions follow by iteration from Lemma 15, as $\Gamma _{\partial ^{%
\text{\textsc{u}}}}\subseteq \Gamma _{\partial }=\emptyset $ obviously
implies $\Gamma _{\partial ^{\text{\textsc{u}}}}=\emptyset $.
\end{proof}

Together with Lemma 5 the latter assertion yields

\begin{corollary}
Plain dag-like \textsc{NM}$_{\rightarrow }$ provability is sound and
complete with respect to minimal propositional logic.
\end{corollary}

Together with Corollary 12 this yields

\begin{conclusion}
A given formula $\rho $ is a tautology in minimal propositional logic iff
there exists a plain dag-like \textsc{NM}$_{\rightarrow }$ proof of $\rho $
whose size and weight are $\mathcal{O}\left( \left| \rho \right| ^{4}\right) 
$ and $\mathcal{O}\left( \left| \rho \right| ^{5}\right) $, respectively.
\end{conclusion}

\subsubsection{Local correctness and complexity of verification}

Our definition of plain dag-like provability via `global' discharging
function (Definition 5) is inappropriate for polytime verification. This is
because `$\alpha $ \emph{is an open (resp. closed) assumption}' refers to
potentially exponential set of threads $\text{\textsc{th\negthinspace }}%
\left( e\!,\varrho ,\partial \right) $ \negthinspace for $e=\left\langle
u,v\right\rangle $ with $\!\ell ^{\text{\textsc{f}}}\left( v\right)
\!=\!\alpha $ in a given plain redag $\partial =\left\langle D,\text{\textsc{%
s}},\ell ^{\text{\textsc{f}}},\ell ^{\text{\textsc{g}}}\text{\textsc{%
\thinspace }}\right\rangle $, thus being merely a NP (resp. coNP) problem,
unless $\partial $\ is a tree. To overcome this obstacle we upgrade basic
(standard) conditions of local correctness of a given (re)dag-like deduction 
$\partial $ (cf. Definition 6) by adding a new labeling function $\ell ^{%
\text{\textsc{d}}}$ that assigns boolean values $0$ or $1$ to all pairs $%
\left( e,\alpha \right) $, where $e=\left\langle u,v\right\rangle $ is an
edge and $\alpha $ an assumption, in $\partial $. Informally, $\ell ^{\text{%
\textsc{d}}}\left( e,\alpha \right) =1$ says that in every $\Theta \in \ $%
\textsc{th}$\left( e\!,\varrho ,\partial \right) $, $\alpha $ is discharged
at $e$, or below, by (sub)occurrences of $\left( \rightarrow I\right) $ with
premise $\alpha $. The corresponding new condition $\bullet $ of the local
correctness in question is shown below, where $K\left( u\right) =\left[
x_{0},\cdots ,x_{k}\right] $ for $x_{0}=u$ and $x_{k}=U\left( u\right) $.

\begin{itemize}
\item  $\ell ^{\text{\textsc{d}}}\left( e,\alpha \right) =1$ iff one of the
following holds.
\end{itemize}

\begin{enumerate}
\item  $u=\varrho $ and $\ell ^{\text{\textsc{f}}}\left( u\right) =\alpha
\rightarrow \ell ^{\text{\textsc{f}}}\left( v\right) $.

\item  $u\neq \varrho $ and

\begin{enumerate}
\item  either $\ell ^{\text{\textsc{f}}}\left( x_{i+1}\right) =\alpha
\rightarrow \ell ^{\text{\textsc{f}}}\left( x_{i}\right) $ holds for some $%
0\leq i<k$,

\item  or $U\left( u\right) \neq \varrho $ and $\underset{w\in \ell ^{\text{%
\textsc{g}}}\left( e\right) }{\prod }\ell ^{\text{\textsc{d}}}\left(
\left\langle w,U\left( u\right) \right\rangle ,\alpha \right) =1$.
\end{enumerate}
\end{enumerate}

Keeping this in mind we can present ``plain'' assertion $\Gamma _{\partial
}=\emptyset $ in a simplified ``encoded'' form $\left( \forall \left\langle
u,v\right\rangle \in \!\text{\textsc{e}}_{0}\left( D\right) \right) \ell ^{%
\text{\textsc{d}}}\left( \left\langle u,v\right\rangle ,\ell ^{\text{\textsc{%
f}}}\left( v\right) \right) =1$. \footnote{{\footnotesize Recall that }$%
K\left( u\right) ${\footnotesize \ are uniquely determined by }$u$%
{\footnotesize \ and }$U\left( u\right) =\varrho ${\footnotesize \ . In
particular this shows that in standard tree-like case the entire
verification is trivial. }}

\begin{definition}
\textsc{NM}$_{\rightarrow }^{\star }$ deductions $\partial $ enriched by $%
\ell ^{\text{\textsc{d}}}$ and satisfying all conditions of the upgraded
local correctness (including $\bullet $), are called \emph{encoded} dag-like 
\textsc{NM}$_{\rightarrow }$ deductions. A given assumption $\alpha $ in an
encoded dag-like \textsc{NM}$_{\rightarrow }$ deduction $\partial $\ of $%
\rho $ is called \emph{closed}\ (or \emph{discharged}) if for every leaf $v$
with $\ell ^{\text{\textsc{f}}}\left( v\right) =\alpha $\ and every edge $%
e=\left\langle u,v\right\rangle $ we have $\ell ^{\text{\textsc{d}}}\left(
e,\alpha \right) =1$. Otherwise $\alpha $ is called \emph{open (or
undischarged)}. Furthermore, as in the case of plain dag-like \textsc{NM}$%
_{\rightarrow }$ deductions, we denote by $\Gamma _{\partial }$ the set of
open assumptions and call $\partial $ an \emph{encoded dag-like} \textsc{NM}$%
_{\rightarrow }$\emph{\ deduction} of $\rho $ from\emph{\ }the assumptions%
\emph{\ }$\Gamma _{\partial }$. If $\Gamma _{\partial }=\emptyset $, then $%
\partial $ is called an \emph{encoded dag-like }\textsc{NM}$_{\rightarrow }$%
\emph{\ proof} of $\rho $.
\end{definition}

\begin{lemma}[plain = encoded]
The notions of plain and encoded deducibility and/or provability are
equivalent, while $\ell ^{\text{\textsc{d}}}$ is uniquely determined by $%
\ell ^{\text{\textsc{f}}}$ and $\ell ^{\text{\textsc{g}}}$. In particular,
any plain dag-like \textsc{NM}$_{\rightarrow }$ proof of $\rho $ can be both
upgraded to and degraded from an encoded dag-like \textsc{NM}$_{\rightarrow
} $ proof of $\rho $ of the same size. Moreover, a statement `$\partial $%
\emph{\ is an encoded dag-like} \textsc{NM}$_{\rightarrow }$ \emph{proof of }%
$\rho $' is verifiable by a TM in $\left\| \partial \right\| $-polynomial
time.
\end{lemma}

\begin{proof}
To prove first two assertions it will suffice to show that for any $%
e=\left\langle u,v\right\rangle \in $ \textsc{e}$_{0}\!\left( D\right) $, $%
\ell ^{\text{\textsc{d}}}\left( e,\ell ^{\text{\textsc{f}}}\left( v\right)
\right) =0$\ holds iff there exists an open thread $\Theta \in \ $\textsc{th}%
$\left( e,\varrho ,\partial \right) $. Actually we observe that a stronger
equivalence stating that for any $e=\!\left\langle u,v\right\rangle \in \!\ $%
\textsc{e}$\left( D\right) $ and $\alpha \in \ell ^{\text{\textsc{f}}}\left( 
\text{\textsc{v}}\left( D\right) \right) $, $\ell ^{\text{\textsc{d}}}\left(
e,\alpha \right) =0$ iff there exists a $\left( \rightarrow \!I\right)
_{\alpha }$-free thread $\Theta \in \ $\textsc{th}$\left( e,\varrho
,\partial \right) $, is provable by induction on $h\left( v\right) $. The
corresponding induction step easily follows from clause 2 (b) of the local
correctness condition $\bullet $. To establish the last assertion we'll
specify standard encoding of an encoded $\partial =\left\langle D,\text{%
\textsc{s}},\ell ^{\text{\textsc{f}}},\ell ^{\text{\textsc{g}}},\ell ^{\text{%
\textsc{d}}}\right\rangle $ in the alphabet of $\mathcal{L}_{\rightarrow }$
extended by $0,1$ and $v_{0},\cdots ,v_{\left| \text{\textsc{v}}\left(
D\right) \right| -1}$ (encoded vertices). Let $N:=\left\{ 0,\cdots ,h\left(
D\right) \!-\!1\right\} $, $V:=\left\{ v_{0},\cdots ,v_{\left| \text{\textsc{%
v}}\left( D\right) \right| -1}\right\} $, $F:=\ell ^{\text{\textsc{f}}%
}\left( \text{\textsc{v}}\left( D\right) \right) $ and consider the
following sets/relations

\begin{equation*}
\begin{array}{c}
H\subseteq \!V\times N,\ E\subseteq \!V^{2},\ D_{1}\subseteq V,\ S\subseteq
V^{2}\cup V^{3},\ K\subseteq \!V^{2}, \\ 
L^{\text{\textsc{f}}}\subseteq \!\!V\times F,\ L^{\text{\textsc{g}}%
}\subseteq E\times V,\ L^{\text{\textsc{d}}}\subseteq E\times F
\end{array}
\end{equation*}
representing respectively 
\begin{equation*}
\begin{array}{c}
H\cong \left\{ \left\langle u,h\left( u\right) \right\rangle \right\} _{u\in 
\text{\textsc{v}}\left( D\right) },\ E\cong \text{\textsc{e}}\left( D\right)
, \\ 
D_{1}\cong \left\{ u\in \text{\textsc{v}}\left( D\right) :\overleftarrow{%
\deg }\left( u\right) =1\right\} ,\ S\cong \left\{ \left\langle
u,z\right\rangle :z\in \text{\textsc{s}}\left( u\right) \right\} _{u\in 
\text{\textsc{v}}\left( D\right) }, \\ 
U\cong \left\{ \left\langle u,U\!\left( u\right) \right\rangle \!\right\}
_{u\in \text{\textsc{v}}\left( D\right) },\ K\cong \left\{ \left\langle
u,x_{i}\right\rangle :K\left( u\right) =\left[ x_{0},\cdots ,x_{k}\right]
,i\leq k\right\} _{u\in \text{\textsc{v}}\left( D\right) }, \\ 
L^{\text{\textsc{f}}}\cong \left\{ \left\langle u,\ell ^{\text{\textsc{f}}%
}\left( u\right) \right\rangle \!\right\} _{u\in \text{\textsc{v}}\left(
D\right) },\ L^{\text{\textsc{g}}}\cong \left\{ \left\langle
e,x\right\rangle \!:x\in \ell ^{\text{\textsc{g}}}\left( e\right) \right\}
_{e\in \text{\textsc{e}}\left( D\right) \!}, \\ 
L^{\text{\textsc{d}}}\cong \left\{ \left\langle e,\alpha \right\rangle :\ell
^{\text{\textsc{d}}}\left( e,\alpha \right) =1\right\} _{e\in \text{\textsc{e%
}}\left( D\right) \!}
\end{array}
\end{equation*}
(cf. Definition 6).\footnote{{\footnotesize For brevity we assume that }$%
v_{0}${\footnotesize \ corresponds to }$\varrho ${\footnotesize . Note that }%
\textsc{c}$\left( u\right) ${\footnotesize \ and }\textsc{p}$\left( u\right) 
${\footnotesize \ are easily parametrizable in }$E${\footnotesize .}} Note
that for any nontrivial $\partial $ we have:

\begin{itemize}
\item  $\left\| V\right\| =$ $\left| V\right| =\left| \partial \right| \leq
\left\| \partial \right\| $,

\item  $\left\| D_{1}\right\| =\left| D_{1}\right| \leq \left| \partial
\right| \leq \left\| \partial \right\| $,

\item  $\left\| H\right\| \leq $ $\left| V\right| \log h\left( D\right) \leq
\left| V\right| \log \left| \text{\textsc{v}}\left( D\right) \right| =\left|
\partial \right| \log \left| \partial \right| <\left\| \partial \right\|
^{2} $,

\item  $\max \left\{ \left\| E\right\| ,\left\| K\right\| \right\} =2\max
\left\{ \left| E\right| ,\left| K\right| \right\} \leq 2\left| \partial
\right| ^{2}<\left\| \partial \right\| ^{3}$,

\item  $\left\| S\right\| =\left| S\right| <2\left| \partial \right|
^{3}\leq \left\| \partial \right\| ^{3}$,

\item  $\left\| L^{\text{\textsc{f}}}\right\| \leq \left| V\right| \times
\mu \left( \partial \right) =\left| \partial \right| \times \mu \left(
\partial \right) \leq \left\| \partial \right\| $,

\item  $\left\| L^{\text{\textsc{g}}}\right\| =\left| L^{\text{\textsc{g}}%
}\right| \leq \left| E\right| \times \left| V\right| =$ $\left| \partial
\right| ^{3}\leq \left\| \partial \right\| ^{3}$,

\item  $\left\| L^{\text{\textsc{d}}}\right\| \leq \left| E\right| \times
\left| F\right| $ $\times \mu \left( \partial \right) \leq \left| \partial
\right| ^{2}\times \phi \left( \partial \right) \times \mu \left( \partial
\right) \leq \left\| \partial \right\| ^{3}$.
\end{itemize}

Hence a tuple $t=\left\langle H,E,L,D_{1},S,K,L^{\text{\textsc{f}}},L^{\text{%
\textsc{g}}},L^{\text{\textsc{d}}}\right\rangle $ can be represented in the
extended language by a string $s$ of the length $\leq $ $\mathcal{O\!}\left(
\left\| \partial \right\| ^{3}\right) $. Having this we observe that
upgraded local correctness of any given encoded redag $\partial
=\left\langle D,\text{\textsc{s}},\ell ^{\text{\textsc{f}}},\ell ^{\text{%
\textsc{g}}},\ell ^{\text{\textsc{d}}}\text{\textsc{\thinspace }}%
\right\rangle $ is a boolean combination of at most $\mathcal{O\!}\left(
\left| \partial \right| ^{9}\right) $\ many elementary equations and queries
over components of $s$. To put it more exactly, the upgraded local
correctness of $\partial $ is the conjunction of the following boolean
assertions $1-24$ for $x$, $y$, $z$,\ $u$, $v$, $w$ and $i$, $j$ and $\alpha 
$, $\beta $ ranging over $V$ and $N$ and $F$, respectively, where we use
abbreviations:

\begin{itemize}
\item  $x\in L:=\fbox{$\underset{y\in V}{\bigwedge }\left\langle
x,y\right\rangle \notin E$},$

\item  $\left\langle x,y\right\rangle \in U:=\fbox{$\left\langle
x,y\right\rangle \in K\wedge y\notin D_{1}$},$

\item  $\left\langle y,x\right\rangle \in \left( \rightarrow I\right)
_{\alpha }:=\fbox{$\underset{\beta \in F}{\bigvee }\left( \left\langle
y,\beta \right\rangle \in L^{\text{\textsc{f}}}\wedge \left\langle x,\alpha
\rightarrow \beta \right\rangle \in L^{\text{\textsc{f}}}\right) $},$

\item  $R\left( u,v,z,\alpha \right) \smallskip :=$

$\fbox{$\left. 
\begin{array}{c}
\underset{x,y\in V}{\bigvee }\left( \left\langle u,x\right\rangle \in
K\wedge \left\langle u,y\right\rangle \in K\wedge \left\langle
y,x\right\rangle \in E\wedge \left\langle y,x\right\rangle \in \left(
\rightarrow I\right) _{\alpha }\right) \\ 
\wedge \underset{w\in V}{\bigwedge }\left( \left\langle \left\langle
u,v\right\rangle ,w\right\rangle \notin L^{\text{\textsc{g}}}\vee
\left\langle \left\langle w,z\right\rangle ,\alpha \right\rangle \notin L^{%
\text{\textsc{d}}}\right)
\end{array}
\right. $}.$
\end{itemize}

\begin{enumerate}
\item  $\fbox{$\underset{i\in N}{\bigvee }\left\langle u,i\right\rangle \in
H $}$

\item  $\fbox{$\left\langle u,0\right\rangle \in H\Leftrightarrow u=\varrho $%
}$

\item  $\fbox{$u\notin L\vee \left\langle u,h\left( D\right)
\!-\!1\right\rangle \in H$}$

\item  $\fbox{$\left\langle u,i\right\rangle \notin E\vee \left\langle
u,j\right\rangle \notin E\vee i=j$}$

\item  $\fbox{$\left\langle u,x\right\rangle \notin E\vee \left\langle
u,y\right\rangle \notin E\vee \left\langle u,i\right\rangle \notin H\vee
\left\langle x,i+1\right\rangle \in H\wedge \left\langle y,i+1\right\rangle
\in H$}$

\item  $\fbox{$\left\langle u,\left\langle x,y\right\rangle \right\rangle
\notin S\vee \left( \left\langle u,x\right\rangle \in E\wedge \left\langle
u,y\right\rangle \in E\right) $}$

\item  $\fbox{$\left\langle u,x\right\rangle \in E\Leftrightarrow
\left\langle u,x\right\rangle \in S\vee \underset{y\in V}{\bigvee }\left(
\left\langle u,\left\langle x,y\right\rangle \right\rangle \in S\vee
\left\langle u,\left\langle y,x\right\rangle \right\rangle \in S\right) $}$

\item  $\fbox{$x\notin D_{1}\vee \underset{v\in V}{\bigvee }\left\langle
v,x\right\rangle \in E$}$

\item  $\fbox{$x\notin D_{1}\vee \left\langle y,x\right\rangle \notin E\vee
\left\langle z,x\right\rangle \notin E\vee y=z$}$

\item  $\fbox{$\left\langle u,u\right\rangle \in K$}$

\item  $\fbox{$u=x\vee \left( \left\langle u,x\right\rangle \in
K\Leftrightarrow \underset{y\in V}{\bigvee }\left( \left\langle
x,y\right\rangle \in E\wedge y\in D_{1}\wedge \left\langle u,y\right\rangle
\in K\right) \right) $}$

\item  $\fbox{$\left\langle \varrho ,\rho \right\rangle \in L^{\text{\textsc{%
f}}}$}$

\item  $\fbox{$\left\langle u,\alpha \right\rangle \notin L^{\text{\textsc{f}%
}}\vee \left\langle u,\beta \right\rangle \notin L^{\text{\textsc{f}}}\vee
\alpha =\beta $}$

\item  $\fbox{$\left\langle u,x\right\rangle \notin S\vee \left\langle
u,\gamma \right\rangle \notin L^{\text{\textsc{f}}}\vee \left\langle
x,\gamma \right\rangle \in L^{\text{\textsc{f}}}\vee \left( \gamma =\alpha
\rightarrow \beta \wedge \left\langle x,\beta \right\rangle \in L^{\text{%
\textsc{f}}}\right) $}$

\item  $\fbox{$\left\langle u,\left\langle x,y\right\rangle \right\rangle
\notin S\vee \left\langle u,\beta \right\rangle \notin L^{\text{\textsc{f}}%
}\vee \underset{a\in F}{\bigvee }\left\langle x,\alpha \right\rangle \in L^{%
\text{\textsc{f}}}\wedge \left\langle y,\alpha \rightarrow \beta
\right\rangle \in L^{\text{\textsc{f}}}$}$

\item  $\fbox{$\left\langle \left\langle u,y\right\rangle ,x\right\rangle
\notin L^{\text{\textsc{g}}}\vee \left\langle u,y\right\rangle \in E$}$

\item  $\fbox{$\left\langle u,y\right\rangle \notin U\vee \left\langle
\left\langle u,v\right\rangle ,x\right\rangle \notin L^{\text{\textsc{g}}%
}\vee \left\langle x,y\right\rangle \in E$}$

\item  $\fbox{$\left\langle u,v\right\rangle \notin E\vee \left\langle
u,y\right\rangle \notin U\vee y=\varrho \vee \underset{x\in V}{\bigvee }%
\left\langle \left\langle u,v\right\rangle ,x\right\rangle \in L^{\text{%
\textsc{g}}}$}$

\item  $\fbox{$\left\langle u,\left\langle v,w\right\rangle \right\rangle
\notin S\vee \left( \left\langle \left\langle u,v\right\rangle
,x\right\rangle \in L^{\text{\textsc{g}}}\Leftrightarrow \left\langle
\left\langle u,w\right\rangle ,x\right\rangle \in L^{\text{\textsc{g}}%
}\right) $}$

\item  $\fbox{$v\notin D_{1}\vee \left( \left\langle \left\langle
u,v\right\rangle ,x\right\rangle \in L^{\text{\textsc{g}}}\Leftrightarrow 
\underset{z\in V}{\bigvee }\left\langle \left\langle v,z\right\rangle
,x\right\rangle \in L^{\text{\textsc{g}}}\right) $}$

\item  $\fbox{$u\in D_{1}\vee \left\langle x,u\right\rangle \notin E\vee 
\underset{v\in V}{\bigvee }\left\langle \left\langle u,v\right\rangle
,x\right\rangle \in L^{\text{\textsc{g}}}$}$

\item  $\fbox{$\left\langle \left\langle u,v\right\rangle ,\alpha
\right\rangle \notin L^{\text{\textsc{d}}}\vee \left\langle u,v\right\rangle
\in E$}$

\item  $\fbox{$\left\langle \left\langle \varrho ,v\right\rangle ,\alpha
\right\rangle \in L^{\text{\textsc{d}}}\Leftrightarrow \left\langle
v,\varrho \right\rangle \in \left( \rightarrow I\right) _{\alpha }$}$

\item  $\fbox{$u=\varrho \vee \left\langle u,z\right\rangle \notin U\vee
\left( \left\langle \left\langle u,v\right\rangle ,\alpha \right\rangle
\notin L^{\text{\textsc{d}}}\Leftrightarrow R\left( u,v,z,\alpha \right)
\right) $}$

It is easily provable by induction on $h\left( D\right) $ that the required
statement `$\partial $\emph{\ is an encoded dag-like} \textsc{NM}$%
_{\rightarrow }$ \emph{proof of }$\rho $' is equivalent to\ universal
conjunction $\left( \overrightarrow{\bigwedge }\right) 1\wedge \cdots \wedge
24\wedge 25$, where by the definition the last condition

\item  \fbox{$\underset{\left\langle u,v\right\rangle \in E}{\bigwedge }%
\underset{\alpha \in F}{\bigwedge }\left\langle v,h\left( D\right)
\!-\!1\right\rangle \notin H\vee \left\langle v,\alpha \right\rangle \notin
L^{\text{\textsc{f}}}\vee \left\langle \left\langle u,v\right\rangle ,\alpha
\right\rangle \in L^{\text{\textsc{d}}}$}

corresponds to $\Gamma _{\partial }=\emptyset $.
\end{enumerate}

The longest conjunct $24$ includes $\leq \mathcal{O\!}\left( \left| \partial
\right| ^{9}\right) $ many equations $\chi =\xi $ for $\chi ,\xi \in X$ and\
queries\ $\chi \in X$, $\chi \notin X$ \ for $X\in \left\{
H,E,L,D_{1},S,K,L^{\text{\textsc{f}}},L^{\text{\textsc{g}}},L^{\text{\textsc{%
d}}}\right\} $, while every query in question is verifiable (say, by binary
search algorithm) by a deterministic TM in $\leq \mathcal{O}\left( \log
\left| s\right| \right) $ time. Hence by any chosen polytime search and
verification algorithm the whole conjunction $\left( \overrightarrow{%
\bigwedge }\right) 1\wedge \cdots \wedge 25$ corresponding to `$\partial $%
\emph{\ is an encoded dag-like} \textsc{NM}$_{\rightarrow }$ \emph{proof of }%
$\rho $' is verifiable by a deterministic TM in $\left\| \partial \right\| $%
-polynomial time.
\end{proof}

\begin{corollary}
$\mathcal{NP=PSPACE}$, and hence $\mathcal{NP=\mathit{co}NP=PSPACE}$.
\end{corollary}

\begin{proof}
Recall that\ the validity problem for both intuitionistic and minimal
propositional logics is PSPACE-complete, cf. \cite{Statman}, \cite{Svejdar}, 
\cite{Savitch}. It will suffice to show that it is a NP problem. So consider
any given $\mathcal{L}_{\rightarrow }$ formula $\rho $. By Conclusion 19, $%
\rho $ is valid in the minimal logic iff there exists an encoded dag-like 
\textsc{NM}$_{\rightarrow }$ proof $\partial $ of $\rho $ of the size $%
\left| \partial \right| =\mathcal{O\!}\left( \left| \rho \right| ^{4}\right) 
$ and weight $\left\| \partial \right\| =\mathcal{O\!}\left( \left| \rho
\right| ^{5}\right) $. Moreover, by Lemma 21, the assertion `$\partial $%
\emph{\ is an encoded dag-like }\textsc{NM}$_{\rightarrow }$\emph{\ proof of 
}$\rho $' is verifiable by a deterministic TM $M$\ in polynomial time with
respect to $\left\| \partial \right\| $, and hence also $\left| \rho \right| 
$. Hence there exists a polytime TM $M$\ such that $\rho $\ is valid in the
minimal logic iff we can ``guess'' an encoded dag-like \textsc{NM}$%
_{\rightarrow }$ proof $\partial $ of the weight $\mathcal{O\!}\left( \left|
\rho \right| ^{5}\right) $\ and confirm its local correctness by $M$ in $%
\left| \rho \right| $-polynomial time. This shows that the underlying
problem of minimal validity belongs to $\mathcal{NP}$, as desired. The rest
follows from Sawitch's theorem \cite{Savitch}.
\end{proof}

\section{Appendix A: proof of Lemma 2 (4)}

A required loose upper bound $\mathrm{ssf}\left( \xi \right) \leq \left(
\left| \xi \right| +1\right) ^{2}$ is proved by induction on $\left| \xi
\right| $, as follows. Recall the recursive clauses 1--3:

\begin{enumerate}
\item  $\mathrm{ssf}\left( p\right) :=1.$

\item  $\mathrm{ssf}\left( p\rightarrow \alpha \right) :=2+\mathrm{ssf}%
\left( \alpha \right) .$

\item  $\mathrm{ssf}\left( \left( \alpha \rightarrow \beta \right)
\rightarrow \gamma \right) :=1+\mathrm{ssf}\left( \alpha \rightarrow \beta
\right) +\mathrm{ssf}\left( \beta \rightarrow \gamma \right) -\mathrm{ssf}%
\left( \beta \right) .$
\end{enumerate}

\begin{itemize}
\item  Basis of induction. Suppose $\left| \xi \right| =0$. Hence $\xi =p$
and $\mathrm{ssf}\left( \xi \right) =1=\left( \left| \xi \right| +1\right)
^{2}$, since $\left| p\right| =0$.

\item  Induction step. Suppose $\left| \xi \right| >0$. Hence $\xi =\alpha
\rightarrow \beta $.

\begin{itemize}
\item  If $\left| \alpha \right| =0$, then $\alpha =p$ and $\mathrm{ssf}%
\left( \xi \right) =2+\mathrm{ssf}\left( \beta \right) \underset{I.H.}{\leq }%
2+\left( \left| \beta \right| +1\right) ^{2}$

$<\left( \left| \beta \right| +2\right) ^{2}=\left( \left| \xi \right|
+1\right) ^{2}.$

\item  Otherwise $\alpha =\gamma \rightarrow \delta $ and $\xi =$ $\left(
\gamma \rightarrow \delta \right) \rightarrow \beta $. If $\left| \delta
\right| =0$, then $\delta =p$ and $\mathrm{ssf}\left( \xi \right) =1+\mathrm{%
ssf}\left( \alpha \right) +\mathrm{ssf}\left( p\rightarrow \beta \right) -%
\mathrm{ssf}\left( p\right) =2+\mathrm{ssf}\left( \alpha \right) +\mathrm{ssf%
}\left( \beta \right) $

$\underset{I.H.}{\leq }2+\left( \left| \alpha \right| +1\right) ^{2}+\left(
\left| \beta \right| +1\right) ^{2}<\left( \left| \alpha \right| +\left|
\beta \right| +1\right) ^{2}=\left( \left| \xi \right| +1\right) ^{2}.$

\item  Otherwise $\delta =\zeta \rightarrow \eta $ and $\xi =$ $\left(
\gamma \rightarrow \left( \zeta \rightarrow \eta \right) \right) \rightarrow
\beta $. If $\left| \eta \right| =0$, then $\eta =p$ and $\mathrm{ssf}\left(
\xi \right) =1+\mathrm{ssf}\left( \alpha \right) +\mathrm{ssf}\left( \left(
\zeta \rightarrow p\right) \rightarrow \beta \right) -\mathrm{ssf}\left(
\zeta \rightarrow p\right) $

$=2+\mathrm{ssf}\left( \alpha \right) +\mathrm{ssf}\left( p\rightarrow \beta
\right) -\mathrm{ssf}\left( p\right) =3+\mathrm{ssf}\left( \alpha \right) +%
\mathrm{ssf}\left( \beta \right) $

$\underset{I.H.}{\leq }3+\left( \left| \alpha \right| +1\right) ^{2}+\left(
\left| \beta \right| +1\right) ^{2}<\left( \left| \alpha \right| +\left|
\beta \right| +1\right) ^{2}=\left( \left| \xi \right| +1\right) ^{2}.$

\item  $\cdots \quad \cdots \quad \cdots \quad \cdots \quad \cdots \quad
\cdots \quad \cdots \quad \cdots \quad \cdots \quad \cdots \quad \cdots
\quad \cdots $

\item  Eventually we arrive at $\alpha =\gamma _{1}\rightarrow \cdots
\rightarrow \gamma _{n}\rightarrow p$ (right-associative) and $\mathrm{ssf}%
\left( \xi \right) =\mathrm{ssf}\left( \alpha \rightarrow \beta \right) =n+1+%
\mathrm{ssf}\left( \alpha \right) +\mathrm{ssf}\left( \beta \right) $

$\underset{I.H.}{\leq }n+1+\left( \left| \alpha \right| +1\right)
^{2}+\left( \left| \beta \right| +1\right) ^{2}<\left( \left| \alpha \right|
+\left| \beta \right| +1\right) ^{2}=\left( \left| \xi \right| +1\right)
^{2}.$
\end{itemize}
\end{itemize}

This completes the proof of Lemma 2 (4).

\section{Appendix B: Gilbert's example}

Let $\partial $ be a following tree-like \textsc{NM}$_{\rightarrow }$ proof
of $\xi \rightarrow s$, where 
\begin{equation*}
\xi =(p_{2}\rightarrow (p_{2}\rightarrow r)\rightarrow q\rightarrow
r)\rightarrow (p_{1}\rightarrow p_{1}\rightarrow (p_{1}\rightarrow
r)\rightarrow q\rightarrow r)\rightarrow s
\end{equation*}
for arbitrary formulas $p_{i}$, $q$, $r$, $s$ of basic minimal language $%
\mathcal{L}_{\rightarrow }$ and $v_{i}$, $u_{i}$, $x_{i}$ are crucial nodes
with formula-labels $\ell ^{\text{\textsc{f}}}\left( v_{1}\right) =\ell ^{%
\text{\textsc{f}}}\left( v_{2}\right) =r$, $\ell ^{\text{\textsc{f}}}\left(
u_{1}\right) =\ell ^{\text{\textsc{f}}}\left( u_{2}\right) =q\rightarrow r$
and $\ell ^{\text{\textsc{f}}}\left( x_{1}\right) =(p_{1}\rightarrow
r)\rightarrow q\rightarrow r\neq \ell ^{\text{\textsc{f}}}\left(
x_{2}\right) =(p_{2}\rightarrow r)\rightarrow q\rightarrow r$.\medskip

\hspace{-20pt} 
\begin{tikzpicture}
  \draw (0.6,2.5) node {
    \AxiomC{$y_1 :[p_1]$}
    \AxiomC{$z_1 :[p_1 \rightarrow r]$}
    \LeftLabel{$(\rightarrow E)$}
    \BinaryInfC{$v_1 : r$}
    \LeftLabel{$(\rightarrow I)$}
    \UnaryInfC{$u_1 : q \rightarrow r$}
    \LeftLabel{$(\rightarrow I)$}
    \UnaryInfC{$x_1 : (p_1 \rightarrow r) \rightarrow q \rightarrow r$}
    \LeftLabel{$(\rightarrow I)$}
    \UnaryInfC{$p_1 \rightarrow (p_1 \rightarrow r) \rightarrow q \rightarrow r$}
    \LeftLabel{$(\rightarrow I)$}
    \UnaryInfC{$p_1 \rightarrow p_1 \rightarrow (p_1 \rightarrow r) \rightarrow q \rightarrow r$}
    \AxiomC{$y_2 :[p_2]$}
    \AxiomC{$z_2 :[p_2 \rightarrow r]$}
    \RightLabel{$(\rightarrow E)$}
    \BinaryInfC{$v_2 : r$}
    \RightLabel{$(\rightarrow I)$}
    \UnaryInfC{$u_2 : q \rightarrow r$}
    \RightLabel{$(\rightarrow I)$}
    \UnaryInfC{$x_2 : (p_2 \rightarrow r) \rightarrow q \rightarrow r$}
    \RightLabel{$(\rightarrow I)$}
    \UnaryInfC{$p_2 \rightarrow (p_2 \rightarrow r) \rightarrow q \rightarrow r$}
    \AxiomC{$[\xi]$}
    \RightLabel{$(\rightarrow E)$}
    \BinaryInfC{$(p_1 \rightarrow p_1 \rightarrow (p_1 \rightarrow r) \rightarrow q \rightarrow r) \rightarrow s$}
    \RightLabel{$(\rightarrow E)$}
    \BinaryInfC{$s$}
    \RightLabel{$(\rightarrow I)$}
    \UnaryInfC{$\xi \rightarrow s$}
    \DisplayProof
  } ;
\end{tikzpicture}

Obviously all five assumptions $p_{1}$, $p_{2}$, $p_{1}\rightarrow r$, $%
p_{2}\rightarrow r$, $\alpha $\ are closed, while $\ell ^{\text{\textsc{g}}%
}\left( e\right) =\emptyset $ for every edge $e$, since $\partial $ is a
tree. Moreover, for any $i\in \left\{ 1,2\right\} $ and assumption $\alpha
\in \{\xi ,p_{i},(p_{i}\rightarrow r)\}$ we have (note that $q$ is not an
assumption in $\partial $):

\begin{enumerate}
\item  $\ell ^{d}(\langle x_{i},u_{i}\rangle ,\alpha )=1\Leftrightarrow $ $%
\alpha \in \{\xi ,p_{i},(p_{i}\rightarrow r)\}$,

\item  $\ell ^{d}(\langle u_{i},v_{i}\rangle ,\alpha )=1\Leftrightarrow
\alpha \in \{\xi ,p_{i},(p_{i}\rightarrow r)\}$.
\end{enumerate}

\vspace{10pt} Hence all assumptions are discharged in $\partial $. Now
consider the compressed dag $\partial ^{\text{\textsc{c}}}$ (for the sake of
brevity we drop redag-like repetitions of $\xi $):\newline

\hspace{-20pt} 
\begin{tikzpicture}
  \draw (0.6,2.5) node {
    \AxiomC{$y_1 :[p_1]$}
    \AxiomC{$z_1 :[p_1 \rightarrow r]$}
    \AxiomC{$y_2 :[p_2]$}
    \AxiomC{$z_2 :[p_2 \rightarrow r]$}
    \RightLabel{$(\rightarrow E)$}
    \QuaternaryInfC{$v : r$}
    \LeftLabel{$(\rightarrow I)$}
    \UnaryInfC{$u : q \rightarrow r$}
    \DisplayProof
  } ;
  \draw [->] (0,1.9) -- (2.5,1.1) ;
  \draw [->] (0,1.9) -- (-2.5,1.1) ;
  \draw (1,0) node {
    \AxiomC{}
    \LeftLabel{$(\rightarrow I)$}
    \UnaryInfC{$x_1 : (p_1 \rightarrow r) \rightarrow q \rightarrow r$}
    \LeftLabel{$(\rightarrow I)$}
    \UnaryInfC{$p_1 \rightarrow (p_1 \rightarrow r) \rightarrow q \rightarrow r$}
    \LeftLabel{$(\rightarrow I)$}
    \UnaryInfC{$p_1 \rightarrow p_1 \rightarrow (p_1 \rightarrow r) \rightarrow q \rightarrow r$}
    \AxiomC{}
    \RightLabel{$(\rightarrow I)$}
    \UnaryInfC{$x_2 : (p_2 \rightarrow r) \rightarrow q \rightarrow r$}
    \RightLabel{$(\rightarrow I)$}
    \UnaryInfC{$p_2 \rightarrow (p_2 \rightarrow r) \rightarrow q \rightarrow r$}
    \AxiomC{$[\xi]$}
    \RightLabel{$(\rightarrow E)$}
    \BinaryInfC{$(p_1 \rightarrow p_1 \rightarrow (p_1 \rightarrow r) \rightarrow q \rightarrow r) \rightarrow s$}
    \RightLabel{$(\rightarrow E)$}
    \BinaryInfC{$s$}
    \RightLabel{$(\rightarrow I)$}
    \UnaryInfC{$\xi \rightarrow s$}
    \DisplayProof
  } ;
\end{tikzpicture}

This time we have:

\begin{enumerate}
\item  $\ell ^{\text{\textsc{g}}}\left( \left\langle v,y_{i}\right\rangle
\right) =\ell ^{\text{\textsc{g}}}\left( \left\langle v,z_{i}\right\rangle
\right) =\left\{ x_{i}\right\} $,

\item  $\ell ^{\text{\textsc{g}}}\left( \left\langle u,v\right\rangle
\right) =\left\{ x_{1},x_{2}\right\} $,

\item  $\ell ^{d}(\langle x_{i},u\rangle ,\alpha )=1\Leftrightarrow $ $%
\alpha \in \{\xi ,p_{i},(p_{i}\rightarrow r)\}$,

\item  $\ell ^{d}(\langle u,v\rangle ,\alpha )=1\Leftrightarrow \alpha =\xi $%
,

\item  $\ell ^{d}(\langle v,y_{i}\rangle ,\alpha )=\ell ^{d}(\langle
x_{i},u\rangle ,\alpha )$,

\item  $\ell ^{d}(\langle v,z_{i}\rangle ,\alpha )=\ell ^{d}(\langle
x_{i},u\rangle ,\alpha )$.
\end{enumerate}

Thus $\ell ^{d}(\langle v,y_{i}\rangle ,\alpha )=\ell ^{d}(\langle
v,z_{i}\rangle ,\alpha )=1$ holds for every $i\in \left\{ 1,2\right\} $ and $%
\alpha \in \{\xi ,p_{i},(p_{i}\rightarrow r)\}$. Hence $\partial ^{\text{%
\textsc{c}}}$ is a dag-like \textsc{NM}$_{\rightarrow }$ proof of $\xi
\rightarrow s$, as expected. As compared to analogous compression from
Example 8, this current $\partial ^{\text{\textsc{c}}}$ allows only two
pairs of maximal deduction threads, which are separated at the lowest mutual
merge point $u$.

\section{Appendix C: Haeusler's example}

Consider the formulas: 1) $\eta =\alpha_{1}\rightarrow\alpha_{2}$, and 2) $%
\sigma_{k}=\alpha_{k-2}\rightarrow(\alpha_{k-1}\rightarrow\alpha_{k})$ for $%
k>2$. Note that $\alpha_{1}\rightarrow\alpha_{n}$ follows from $\eta ,\sigma
_{3},\ldots ,\sigma_{n}$ and\ the size of standard tree-like normal proof of
this statement exceeds $Fibonnacci(n)$. For $n=5$ we have the derivation
that is shown in Fig. ~\ref{huge}

Generally, for each $5\leq n$ we arrive at 
\begin{prooftree}
\def\defaultHypSeparation{\hskip .05in}
\def\ScoreOverang{1pt}
\AxiomC{$[\alpha_1]$}
\noLine
\UnaryInfC{$\eta$}
\noLine
\UnaryInfC{$\sigma_3,\ldots,\sigma_{n-1}$}
\noLine
\UnaryInfC{$\Pi_{n-1}$}
\noLine
\UnaryInfC{$\alpha_{n-1}$}
\AxiomC{$[\alpha_1]$}
\noLine
\UnaryInfC{$\eta$}
\noLine
\UnaryInfC{$\sigma_3,\ldots,\sigma_{n-2}$}
\noLine
\UnaryInfC{$\Pi_{n-2}$}
\noLine
\UnaryInfC{$\alpha_{k-2}$}
\AxiomC{$\alpha_{n-2}\imply (\alpha_{n-1}\imply \alpha_n)$}
\BinaryInfC{$\alpha_{n-1}\imply \alpha_n$}
\BinaryInfC{$\alpha_n$}
\UnaryInfC{$\alpha_1\imply \alpha_n$}
\end{prooftree}
\begin{eqnarray*}
l(\Pi _{2}) &=&1 \\
l(\Pi _{3}) &=&l(\Pi _{2})+1 \\
l(\Pi _{k}) &=&l(\Pi _{k-2})+l(\Pi _{k-1})+2
\end{eqnarray*}
\begin{equation*}
Fibonacci(n)\leq l(\Pi _{n})
\end{equation*}

\textbf{Towards polynomial representation}. \newline
Using (re)dags we compress our tree-like proofs by merging distinct
occurrences of identical formulas $\alpha _{3}$, $\alpha _{2}$, $\alpha _{1}$
as shown in Fig.~\ref{one},~\ref{two} and~\ref{three}.

\begin{figure*}[tbp]
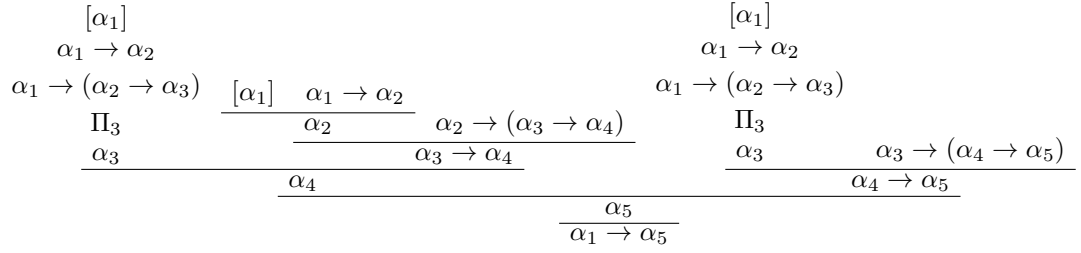

\begin{prooftree}
\def\defaultHypSeparation{\hskip .05in}
\def\ScoreOverang{1pt}
\AxiomC{$[\alpha_1]$}
\noLine
\UnaryInfC{$\alpha_1\imply \alpha_2$}
\noLine
\UnaryInfC{$\alpha_1\imply (\alpha_2\imply \alpha_3)$}
\noLine
\UnaryInfC{$\Pi_3$}
\noLine
\UnaryInfC{$\alpha_3$}
\AxiomC{$[\alpha_1]$}
\AxiomC{$\alpha_1\imply \alpha_2$}
\BinaryInfC{$\alpha_2$}
\AxiomC{$\alpha_2\imply(\alpha_3\imply \alpha_4)$}
\BinaryInfC{$\alpha_3\imply \alpha_4$}
\BinaryInfC{$\alpha_4$}
\AxiomC{$[\alpha_1]$}
\noLine
\UnaryInfC{$\alpha_1\imply \alpha_2$}
\noLine
\UnaryInfC{$\alpha_1\imply(\alpha_2\imply \alpha_3)$}
\noLine
\UnaryInfC{$\Pi_3$}
\noLine
\UnaryInfC{$\alpha_3$}
\AxiomC{$\alpha_3\imply(\alpha_4\imply \alpha_5)$}
\BinaryInfC{$\alpha_4\imply \alpha_5$}
\BinaryInfC{$\alpha_5$}
\UnaryInfC{$\alpha_1\imply \alpha_5$}
\end{prooftree}
\caption{A huge ND proof}
\label{huge}
\end{figure*}

{\tiny 
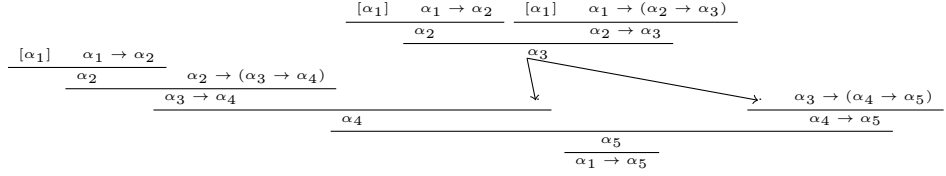
\begin{figure*}[tbp]
{\tiny 
\begin{prooftree}
\def\defaultHypSeparation{\hskip .05in}
\def\ScoreOverang{1pt}
\AxiomC{$[\alpha_1]$}
\AxiomC{$\alpha_1\imply \alpha_2$}
\BinaryInfC{$\alpha_2$}
\AxiomC{$\alpha_2\imply(\alpha_3\imply \alpha_4)$}
\BinaryInfC{$\alpha_3\imply \alpha_4$}
\AxiomC{$[\alpha_1]$}
\AxiomC{$\alpha_1\imply \alpha_2$}
\BinaryInfC{$\alpha_2$}
\AxiomC{$[\alpha_1]$}
\AxiomC{$\alpha_1\imply (\alpha_2\imply \alpha_3)$}
\BinaryInfC{$\alpha_2\imply \alpha_3$}
\BinaryInfC{\tikzmark{x}$\alpha_3$}
\noLine
\UnaryInfC{\;}
\noLine
\UnaryInfC{\;}
\noLine
\UnaryInfC{\tikzmark{y1}.}
\BinaryInfC{$\alpha_4$}
\AxiomC{\tikzmark{y2}.}
\AxiomC{$\alpha_3\imply(\alpha_4\imply \alpha_5)$}
\BinaryInfC{$\alpha_4\imply \alpha_5$}
\BinaryInfC{$\alpha_5$}
\UnaryInfC{$\alpha_1\imply \alpha_5$}
\end{prooftree}
\begin{tikzpicture}[remember picture, overlay]
  \draw [->] ({pic cs:x}) +(0,-.2\baselineskip) to ({pic cs:y1});
  \draw [->] ({pic cs:x}) +(0,-.2\baselineskip) to ({pic cs:y2});
\end{tikzpicture}
}
\caption{Horizontal compression (1)}
\label{one}
\end{figure*}
} {\tiny 
\begin{figure*}[tbp]
{\tiny 
\begin{prooftree}
\def\defaultHypSeparation{\hskip .05in}
\def\ScoreOverang{1pt}
\AxiomC{\tikzmark{we1}.}
\AxiomC{$\alpha_2\imply(\alpha_3\imply \alpha_4)$}
\BinaryInfC{$\alpha_3\imply \alpha_4$}
\AxiomC{$[\alpha_1]$}
\AxiomC{$\alpha_1\imply \alpha_2$}
\BinaryInfC{\tikzmark{ze}$\alpha_2$}
\noLine
\UnaryInfC{\;}
\noLine
\UnaryInfC{\;}
\noLine
\UnaryInfC{\tikzmark{we2}.}
\AxiomC{$[\alpha_1]$}
\AxiomC{$\alpha_1\imply (\alpha_2\imply \alpha_3)$}
\BinaryInfC{$\alpha_2\imply \alpha_3$}
\BinaryInfC{\tikzmark{xee}$\alpha_3$}
\noLine
\UnaryInfC{\;}
\noLine
\UnaryInfC{\;}
\noLine
\UnaryInfC{\tikzmark{yee1}.}
\BinaryInfC{$\alpha_4$}
\AxiomC{\tikzmark{yee2}.}
\AxiomC{$\alpha_3\imply(\alpha_4\imply \alpha_5)$}
\BinaryInfC{$\alpha_4\imply \alpha_5$}
\BinaryInfC{$\alpha_5$}
\UnaryInfC{$\alpha_1\imply \alpha_5$}
\end{prooftree}
\begin{tikzpicture}[remember picture, overlay]
  \draw [->] ({pic cs:xee}) +(0,-.2\baselineskip) to ({pic cs:yee1});
  \draw [->] ({pic cs:xee}) +(0,-.2\baselineskip) to ({pic cs:yee2});
  \draw [->] ({pic cs:ze}) +(0,-.2\baselineskip) to ({pic cs:we1});
  \draw [->] ({pic cs:ze}) +(0,-.2\baselineskip) to ({pic cs:we2});
\end{tikzpicture}
}
\caption{Horizontal compression (2)}
\label{two}
\end{figure*}
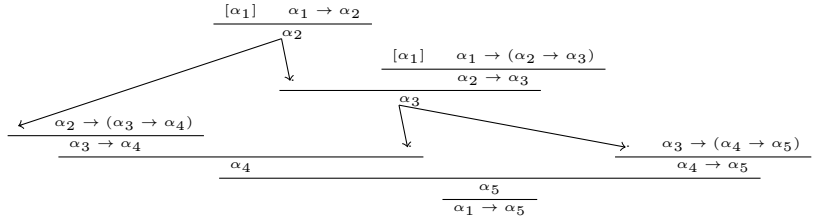
} {\tiny 
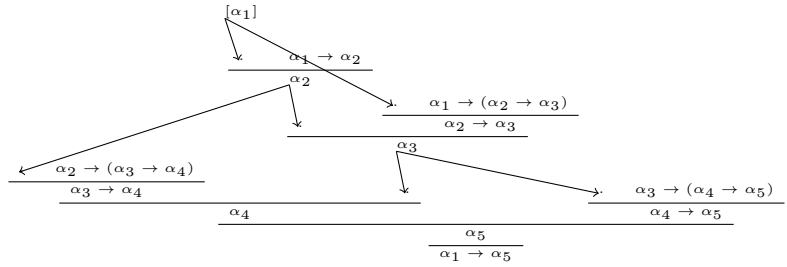
\begin{figure*}[tbp]
{\tiny 
\begin{prooftree}
\def\defaultHypSeparation{\hskip .05in}
\def\ScoreOverang{1pt}
\AxiomC{\tikzmark{w1}.}
\AxiomC{$\alpha_2\imply(\alpha_3\imply \alpha_4)$}
\BinaryInfC{$\alpha_3\imply \alpha_4$}
\AxiomC{\tikzmark{n}$[\alpha_1]$}
\noLine
\UnaryInfC{\;}
\noLine
\UnaryInfC{\;}
\noLine
\UnaryInfC{\tikzmark{m2}.}
\AxiomC{$\alpha_1\imply \alpha_2$}
\BinaryInfC{\tikzmark{z}$\alpha_2$}
\noLine
\UnaryInfC{\;}
\noLine
\UnaryInfC{\;}
\noLine
\UnaryInfC{\tikzmark{w2}.}
\AxiomC{\tikzmark{m1}.}
\AxiomC{$\alpha_1\imply (\alpha_2\imply \alpha_3)$}
\BinaryInfC{$\alpha_2\imply \alpha_3$}
\BinaryInfC{\tikzmark{xe}$\alpha_3$}
\noLine
\UnaryInfC{\;}
\noLine
\UnaryInfC{\;}
\noLine
\UnaryInfC{\tikzmark{ye1}.}
\BinaryInfC{$\alpha_4$}
\AxiomC{\tikzmark{ye2}.}
\AxiomC{$\alpha_3\imply(\alpha_4\imply \alpha_5)$}
\BinaryInfC{$\alpha_4\imply \alpha_5$}
\BinaryInfC{$\alpha_5$}
\UnaryInfC{$\alpha_1\imply \alpha_5$}
\end{prooftree}
\begin{tikzpicture}[remember picture, overlay]
  \draw [->] ({pic cs:xe}) +(0,-.2\baselineskip) to ({pic cs:ye1});
  \draw [->] ({pic cs:xe}) +(0,-.2\baselineskip) to ({pic cs:ye2});
  \draw [->] ({pic cs:z}) +(0,-.2\baselineskip) to ({pic cs:w1});
  \draw [->] ({pic cs:z}) +(0,-.2\baselineskip) to ({pic cs:w2});
  \draw [->] ({pic cs:n}) +(0,-.2\baselineskip) to ({pic cs:m1});
  \draw [->] ({pic cs:n}) +(0,-.2\baselineskip) to ({pic cs:m2});
\end{tikzpicture}
}
\caption{Horizontal compression (3)}
\label{three}
\end{figure*}
}

This procedure results in the plain dag-like proof shown in figure~\ref
{encode1} (afterwards encoded in Fig.~\ref{encode2}, see below), where we
assume that for every non-leaf node $x$, \textsc{s}$\left( x\right) $
contains all $x$'s\ children available, while for every downward-branching
node $u$, every $\ell ^{\text{\textsc{g}}}\left( \left\langle
u,v\right\rangle \right) $ contains all parents of $u$. Obviously this
dag-like deduction is smaller than its tree-like original. Generally, we
obtain dag-like deductions (not encoded yet) of $\alpha _{1}\rightarrow
\alpha _{n}$ , whose size is smaller than $\sum_{i=1,n}i$, i.e. $O(n^{2})$.
The corresponding encoded dag-like deduction of $\alpha _{1}\rightarrow
\alpha _{5}$ is shown in Fig.~\ref{encode2}, where a string of bits $%
b_{1}b_{2}\ldots b_{5}$ represents the discharging function $\ell ^{\text{%
\textrm{d}}}$. Namely, for any $e$ we let $\ell ^{d}(e,\xi _{i})=b_{i}$ iff $%
\xi _{i}$ is the $i^{th}$ assumption with respect to lexicographical order $%
\alpha _{1}\prec \alpha _{1}\rightarrow \alpha _{2}\prec \alpha
_{1}\rightarrow \left( \alpha _{2}\rightarrow \alpha _{3}\right) \prec
\alpha _{2}\rightarrow \left( \alpha _{3}\rightarrow \alpha _{4}\right)
\prec \alpha _{3}\rightarrow \left( \alpha _{4}\rightarrow \alpha
_{5}\right) $. Actually we always arrive at $b_{1}b_{2}\ldots b_{5}=10000$,
as $\xi _{i}=\alpha _{1}$ is the only closed assumption and it is discharged
by the root inference $\left( \rightarrow \text{\textsc{I}}\right) $. (For
brevity we don't expose labeling functions $\ell ^{\text{\textsc{f}}}$ and $%
\ell ^{\text{\textsc{G}}}$.)

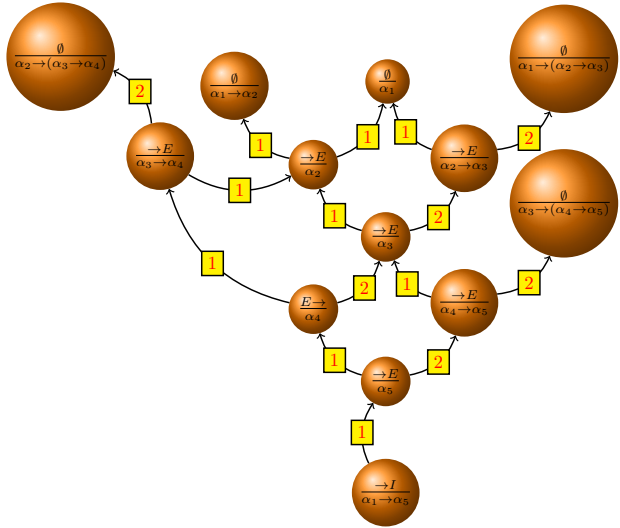
\begin{figure*}[tbp]
\begin{center}
\scalebox{.7}{
\begin{tikzpicture}[node distance = 2cm]
  \useasboundingbox (-4,-1) rectangle (11,11); 
  \tikzset{VertexStyle/.style = {shape          = circle,
                                 ball color     = orange,
                                 text           = black,
                                 inner sep      = 2pt,
                                 outer sep      = 0pt,
                                 minimum size   = 24 pt}}
  \tikzset{EdgeStyle/.style   = {thick,
                                 ->, bend left}}
  \tikzset{LabelStyle/.style =   {draw,
                                  fill           = yellow,
                                  text           = red}}
  \SetGraphUnit{5}
  \node[VertexStyle](A1ImpA5){$\frac{\imply I}{\alpha_1\imply \alpha_5}$};
  \node[VertexStyle,above= 1 cm of A1ImpA5](A5){$\frac{\imply E}{\alpha_5}$};
  \node[VertexStyle,above left= 1 cm of A5](A4){$\frac{E\imply}{\alpha_4}$};
  \node[VertexStyle,above right= 1 cm of A5](A4ImpA5){$\frac{\imply E}{\alpha_4\imply \alpha_5}$};
  \node[VertexStyle,above right= 1 cm of A4ImpA5](A3ImpA5){$\frac{\emptyset}{\alpha_3\imply(\alpha_4\imply \alpha_5)}$};
  \node[VertexStyle,above right= 1 cm of A4](A3){$\frac{\imply E}{\alpha_3}$};
  \node[VertexStyle,above right= 1 cm of A3](A2ImpA3){$\frac{\imply E}{\alpha_2\imply \alpha_3}$};
  \node[VertexStyle,above right= 1 cm of A2ImpA3](A1ImpA3){$\frac{\emptyset}{\alpha_1\imply(\alpha_2\imply \alpha_3)}$};
  \node[VertexStyle,above left= 1 cm of A2ImpA3](A1){$\frac{\emptyset}{\alpha_1}$};
  \node[VertexStyle,above left= 1 cm of A3](A2){$\frac{\imply E}{\alpha_2}$};
  \node[VertexStyle,above left= 1 cm of A2](A1ImpA2){$\frac{\emptyset}{\alpha_1\imply \alpha_2}$};
  \node[VertexStyle,above left= 3 cm of A4](A3ImpA4){$\frac{\imply E}{\alpha_3\imply \alpha_4}$};
  \node[VertexStyle,above left= 1 cm of A3ImpA4](A2ImpA4){$\frac{\emptyset}{\alpha_2\imply(\alpha_3\imply \alpha_4)}$};
  \draw[EdgeStyle](A1ImpA5) to node[LabelStyle]{1} (A5);
  \draw[EdgeStyle](A5) to node[LabelStyle]{1} (A4);
  \tikzset{EdgeStyle/.append style = {bend right}}
  \draw[EdgeStyle](A5) to node[LabelStyle]{2} (A4ImpA5);
  \draw[EdgeStyle](A4ImpA5) to node[LabelStyle]{2} (A3ImpA5);
  \draw[EdgeStyle](A4) to node[LabelStyle]{2} (A3);
  \tikzset{EdgeStyle/.append style = {bend left}};
  \draw[EdgeStyle](A4ImpA5) to node[LabelStyle]{1} (A3);
  \draw[EdgeStyle](A4) to node[LabelStyle]{1} (A3ImpA4);
  \tikzset{EdgeStyle/.append style = {bend right}}
  \draw[EdgeStyle](A3ImpA4) to node[LabelStyle]{2} (A2ImpA4);
  \draw[EdgeStyle](A3) to node[LabelStyle]{2} (A2ImpA3);
  \draw[EdgeStyle](A2ImpA3) to node[LabelStyle]{2} (A1ImpA3);
  \draw[EdgeStyle](A2) to node[LabelStyle]{1} (A1);
  \draw[EdgeStyle](A3ImpA4) to node[LabelStyle]{1} (A2);
  \tikzset{EdgeStyle/.append style = {bend left}};
  \draw[EdgeStyle](A3) to node[LabelStyle]{1} (A2);
  \draw[EdgeStyle](A2) to node[LabelStyle]{1} (A1ImpA2);
  \draw[EdgeStyle](A2ImpA3) to node[LabelStyle]{1} (A1);
\end{tikzpicture}
}
\end{center}
\caption{Encoding the dag-like proof (1)}
\label{encode1}
\end{figure*}

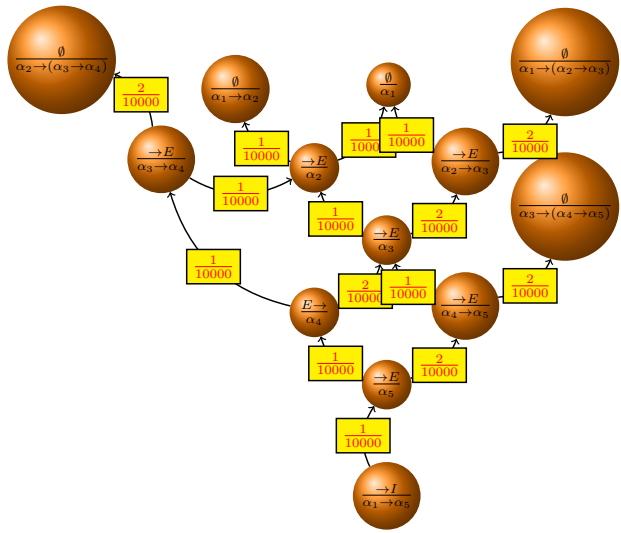
\begin{figure*}[tbp]
\begin{center}
\scalebox{.7}{
\begin{tikzpicture}[node distance = 2cm]
  \useasboundingbox (-4,-1) rectangle (11,11); 
  \tikzset{VertexStyle/.style = {shape          = circle,
                                 ball color     = orange,
                                 text           = black,
                                 inner sep      = 2pt,
                                 outer sep      = 0pt,
                                 minimum size   = 24 pt}}
  \tikzset{EdgeStyle/.style   = {thick,
                                 ->, bend left}}
  \tikzset{LabelStyle/.style =   {draw,
                                  fill           = yellow,
                                  text           = red}}
  \SetGraphUnit{5}
  \node[VertexStyle](A1ImpA5){$\frac{\imply I}{\alpha_1\imply \alpha_5}$};
  \node[VertexStyle,above= 1 cm of A1ImpA5](A5){$\frac{\imply E}{\alpha_5}$};
  \node[VertexStyle,above left= 1 cm of A5](A4){$\frac{E\imply}{\alpha_4}$};
  \node[VertexStyle,above right= 1 cm of A5](A4ImpA5){$\frac{\imply E}{\alpha_4\imply \alpha_5}$};
  \node[VertexStyle,above right= 1 cm of A4ImpA5](A3ImpA5){$\frac{\emptyset}{\alpha_3\imply(\alpha_4\imply \alpha_5)}$};
  \node[VertexStyle,above right= 1 cm of A4](A3){$\frac{\imply E}{\alpha_3}$};
  \node[VertexStyle,above right= 1 cm of A3](A2ImpA3){$\frac{\imply E}{\alpha_2\imply \alpha_3}$};
  \node[VertexStyle,above right= 1 cm of A2ImpA3](A1ImpA3){$\frac{\emptyset}{\alpha_1\imply(\alpha_2\imply \alpha_3)}$};
  \node[VertexStyle,above left= 1 cm of A2ImpA3](A1){$\frac{\emptyset}{\alpha_1}$};
  \node[VertexStyle,above left= 1 cm of A3](A2){$\frac{\imply E}{\alpha_2}$};
  \node[VertexStyle,above left= 1 cm of A2](A1ImpA2){$\frac{\emptyset}{\alpha_1\imply \alpha_2}$};
  \node[VertexStyle,above left= 3 cm of A4](A3ImpA4){$\frac{\imply E}{\alpha_3\imply \alpha_4}$};
  \node[VertexStyle,above left= 1 cm of A3ImpA4](A2ImpA4){$\frac{\emptyset}{\alpha_2\imply(\alpha_3\imply \alpha_4)}$};
  \draw[EdgeStyle](A1ImpA5) to node[LabelStyle]{$\frac{1}{10000}$} (A5);
  \draw[EdgeStyle](A5) to node[LabelStyle]{$\frac{1}{10000}$} (A4);
  \tikzset{EdgeStyle/.append style = {bend right}}
  \draw[EdgeStyle](A5) to node[LabelStyle]{$\frac{2}{10000}$} (A4ImpA5);
  \draw[EdgeStyle](A4ImpA5) to node[LabelStyle]{$\frac{2}{10000}$} (A3ImpA5);
  \draw[EdgeStyle](A4) to node[LabelStyle]{$\frac{2}{10000}$} (A3);
  \tikzset{EdgeStyle/.append style = {bend left}};
  \draw[EdgeStyle](A4ImpA5) to node[LabelStyle]{$\frac{1}{10000}$} (A3);
  \draw[EdgeStyle](A4) to node[LabelStyle]{$\frac{1}{10000}$} (A3ImpA4);
  \tikzset{EdgeStyle/.append style = {bend right}}
  \draw[EdgeStyle](A3ImpA4) to node[LabelStyle]{$\frac{2}{10000}$} (A2ImpA4);
  \draw[EdgeStyle](A3) to node[LabelStyle]{$\frac{2}{10000}$} (A2ImpA3);
  \draw[EdgeStyle](A2ImpA3) to node[LabelStyle]{$\frac{2}{10000}$} (A1ImpA3);
  \draw[EdgeStyle](A2) to node[LabelStyle]{$\frac{1}{10000}$} (A1);
  \draw[EdgeStyle](A3ImpA4) to node[LabelStyle]{$\frac{1}{10000}$} (A2);
  \tikzset{EdgeStyle/.append style = {bend left}};
  \draw[EdgeStyle](A3) to node[LabelStyle]{$\frac{1}{10000}$} (A2);
  \draw[EdgeStyle](A2) to node[LabelStyle]{$\frac{1}{10000}$} (A1ImpA2);
  \draw[EdgeStyle](A2ImpA3) to node[LabelStyle]{$\frac{1}{10000}$} (A1);
\end{tikzpicture}
}
\end{center}
\caption{Horizontal compression (2)}
\label{encode2}
\end{figure*}


\begin{thebibliography}{99}
\bibitem{JMS}  L. Gordeev, E. H. Haeusler, V. G. da Costa, \emph{Proof
compressions with circuit-structured substitutions}, J. Math. Sci. 158(5):
645--658 (2009)

\bibitem{IGPL}  L. Gordeev, E. H. Haeusler, L. C. Pereira, \emph{%
Propositional proof compressions and DNF logic}, Logic Journal of the IGPL
19(1): 62-86 (2011)

\bibitem{Gor}  L. Gordeev, \emph{Basic dag compressions}, Manuscript (2015)

\bibitem{Hudelmaier}  J. Hudelmaier, \emph{An }$O\left( n\log n\right) $%
\emph{-space decision procedure for intuitionistic propositional logic}, J.
Logic Computat. (3): 1--13 (1993)

\bibitem{Johansson}  I. Johansson, \emph{Der Minimalkalk\"{u}l, ein
reduzierter intuitionistischer Formalismus}, Compositio Mathematica (4):
119--136 (1936)

\bibitem{Prawitz}  D. Prawitz, \textbf{Natural deduction: a
proof-theoretical study}. Almqvist \& Wiksell, 1965

\bibitem{Marcela}  M. Quispe-Cruz, E. H. Haeusler, L. Gordeev, \emph{%
Proof-graphs for minimal implicational logic}, Proc. DCM: 16-29 (2014)

\bibitem{Savitch}  W. Savitch, \emph{Relationships between nondeterministic
and deterministic tape complexities}, J. of Computer and System Sciences
(4): 177--192 (1970)

\bibitem{Statman}  R. Statman, \emph{Intuitionistic propositional logic is
polynomial-space complete}, Theor. Comp. Sci. (9): 67--72 (1979)

\bibitem{Svejdar}  V. \^{S}vejdar, \emph{On the polynomial-space
completeness of intuitionistic propositional logic}, Archive for Math. Logic
(42): 711--716 (2003)
\end{thebibliography}
\end{document}